\documentclass[11pt]{article} 
\usepackage{algorithm}
\usepackage{algorithmic}
\usepackage{epsfig}
\usepackage{subfigure}
\usepackage{amsfonts}
\usepackage{amssymb}
\usepackage{amsmath}
\usepackage{cite}
\usepackage{color}
\usepackage{graphicx}
\usepackage{graphics}
\usepackage{latexsym}
\usepackage{mathrsfs}
\usepackage{psfrag}
\usepackage{times} 
\usepackage{xspace}
\usepackage{setspace}

\def\qed{\hbox{\rlap{$\sqcap$}$\sqcup$}}
\newtheorem{theorem}{Theorem}[section]
\newtheorem{lemma}{Lemma}[section]
\newtheorem{property}{Property}[section]
\newtheorem{cor}{Corollary}[theorem]
\newtheorem{definition}{Definition}[section]
\newenvironment{proof}{\par\noindent{\bf Proof:}}{\mbox{}\hfill$\qed$\\}
\newcounter{rem}
\newcommand{\ignore}[1]{ }

\def\PS{{\cal P}}
\def\cTray{{\cal T}_R}
\def\cE{{\cal E}}
\def\cC{{\cal C}}
\def\cK{{\cal K}}

\def\cT{{\cal T}_R}
\def\cRay{{\cal R}}
\def\seg{{\cal CS}}

\begin{document}

\title{\underbar \bf Approximate Shortest Path through a Weighted Planar Subdivision} 
\author{Rajasekhar Inkulu \thanks{Department of Computer Science, Indian Institute of Technology, Guwahati, India.  E-mail: rinkulu@iitg.ac.in} \and Sanjiv Kapoor \thanks{Department of Computer Science, Illinois Institute of Technology, Chicago, USA.  E-mail: kapoor@iit.edu}}
\date{}
\maketitle

\begin{abstract}
This paper presents an approximation algorithm for finding a shortest path between two points $s$ and $t$ in a weighted planar subdivision $\PS$.
Each face $f$ of $\PS$ is associated with a weight $w_f$, and the cost of travel along a line segment on $f$ is $w_f$ multiplied by the Euclidean norm of that line segment. 
The cost of a path which traverses across several faces of the subdivision is the sum of the costs of travel along each face.
Our algorithm progreeses the discretized shortest path wavefront from source $s$, and takes polynomial time in finding an $\epsilon$-approximate shortest path.
\end{abstract}

\section{Introduction}
\label{sect:weispintro}

Let $\PS$ be a planar subdivision, specified by $n$ vertices and $O(n)$ faces.
We assume, w.l.o.g., that the faces are triangles.
Each face $f$ has a weight, $\alpha_f \in {\textrm Z}^+$, associated with it.  
This weight specifies the cost per unit distance of traversing the  face $f$.  
For simplicity, we assume that the cost of traversing a  unit distance along an edge $e$, say $\alpha_e$, which bounds faces $f_1$ and $f_2$ is min\{$\alpha_{f_1}, \alpha_{f_2}$\}.  
The cost of a straight line path between two points $p$ and $q$ on face $f$ is $\alpha_f \Vert pq \Vert$.
And the weighted length of a path passing through a number of faces is the sum of the weighted lengths of its sub-paths through each face.
Let $s$ and $t$, respectively, be the source and destination points on $\PS$.
The weighted shortest path problem is to find a path having the minimum weighted length from $s$ to $t$.
The path we need to find is geodesic in the sense that it is locally optimal and cannot, therefore, be shortened by slight perturbations.  
An optimal path is a geodesic path that is globally optimal.
A corresponding problem can also be defined when the domain, $\PS$ is the surface of a polyhedron.
The applications of this problem are rich and varied, ranging from motion planning, to graphics, injection molding and computer-assisted surgery \cite{Mitchell91, Aleks05}.

While there has been previous work on this problem (see Table \ref{tab:timecomp}), a polynomial solution which tackled the problem in a systematic way was designed by Mitchell and Papadimitriou \cite{Mitchell91}.
Their algorithm uses Snell's Law and a continuous Dijkstra method to give an optimal-path map for any given source point $s$.  
Because of numerical precision issues, the algorithm \cite{Mitchell91} finds an $\epsilon$-approximation to an optimal path.
In this paper too, we consider $\epsilon$-approximate paths i.e., the weighted approximate shortest path from $s$ to $t$ is of length at most $(1+\epsilon)$ times the weighted distance between $s$ to $t$ of an optimal path.

Given the higher order polynomial solution for this problem, approximation algorithms are provided by Agarwal et al. \cite{Agarwal02}, Cheng S.-W. et al. \cite{Cheng08}, Aleksandrov et al. \cite{Aleks05},  and Reif et al. \cite{Reif00, Sun01} whose time complexities rely on geometric parameters that are outside the logarithm.
However, the time complexities of these algorithms are pseudo-polynomial (running time is polynomial in the numeric value of the input i.e., exponential in the length of the number of bits in the input).
The algorithm by \cite{Cheng08} has a smaller dependence on $n$ as compared with \cite{Mitchell91} but a higher dependece on $\epsilon$ and $\rho$ (defined in Table \ref{tab:timecomp}).
Previous algorithms by Lanthier et al. \cite{Lanthier01} and Aleksandrov et al. \cite{Aleks05} took a discretization approach by introducing Steiner points on each edge.
A discrete graph is constructed by adding edges connecting Steiner points in the same triangular region and an optimal path is computed in the resulting discrete graph using Dijkstra's algorithm.
To avoid  incurring  high time complexity, they have used a subgraph of the complete graph in each triangular region.

In this paper, we design an efficient polynomial time approximation algorithm for finding $\epsilon$-approximate paths in weighted planar subdivisions.
Previous schemes for this problem either discretize the space using Steiner vertices \cite{Aleks05} or partition the space around a point into cones \cite{Mata97} and process all the discrete objects created.
Our approach discretizes the wavefront in the continuous-Dijkstra's approach by initiating a set of rays from $s$, and the rays are generated uniformly around a source point.
Since the approach of maintaining the wavefront exactly is complicated, we do not use it in this paper.
Rather, in our approach, we utilize only a subset of rays to guide the progress of the wavefront as it expands from the source. 
\ignore{
These subset of rays are used to define the range of rays intersecting an edge of the polygonal domain, and are updated as the wavefront progresses.
when any ray $r$ is incident on an edge $e$ at an angle less than the critical angle corresponding to $e$, we progress $r$ according to Snell's law;
when any ray $r$ is incident on an edge $e$ at an angle greater than or equal to critical angle corresponding to $e$, rays are generated along the critical segment corresponding to that incidence.

To account for the divergence of the rays, we utilize another special set of rays, known as Steiner rays,  generated from critical points of entry on edges.
The Steiner rays are motivated as follows: A pair of adjacent rays when progressing along a sequence of weighted faces can diverge non-uniformly, especially at angles close to the critical reflection angle. 
To establish the approximation, the discretization must generate discrete rays with sufficient density in every part of the domain; thus we fill the gaps by generating the Steiner rays from these critical points, instead of generating a more denser set of rays from the source itself.
This helps in reducing the time complexity of discretization.
}
When two rays together yield an $\epsilon$-approximate shortest distance from $s$ to a vertex $v$, we initiate another discretized wavefront from $v$.
Thus our propagation processes discrete wavefronts from every vertex and can also be used to solve the all-pairs shortest path problem in the domain.
This method is also applicable in finding geodesic shortest path between two given points such that the shortest path lies on the surface of a 2-manifold whose faces are associated with positive weights.
Our solution uses several of the Lemmas and Theorems proved in Mitchell and Papadimitriou \cite{Mitchell91}, and extends their solution to obtain a better bound. 

The current best known algorithm \cite{Aleks05} to solve this problem requires $O(K\frac{n}{\sqrt{\epsilon}}\lg{(\frac{NW}{w})}$ $\lg{\frac{n}{\epsilon}}\lg{\frac{1}{\epsilon}})$ time, where $K$ captures geometric parameters and the weights of the faces.
This paper describes an algorithm with $O(n^6\lg(\frac{1}{\epsilon'})+n^4\lg(\frac{W}{w\epsilon'}))$ time complexity, where $\epsilon' = \frac{\epsilon}{n(n^2+LW)}$, $n$ is the number of vertices, $w$ and $W$ are the minimum and maximum weights among all the weights associated with the faces of $\PS$, and $L$ is the maximum length among all edge lengths of $\PS$.
The time complexity of our algorithm is better than the known polynomial time algorithms, and it does not contain terms that rely on geometric parameters outside the logarithm.
Table \ref{tab:timecomp} compares time complexity of our result with the previous work.
We assume that all parameters of the problem are specified by integers.
All the weights are assumed to be greater than or equal to one.
Note, that as established in \cite{Mitchell91}, the number of points required to define the interaction of the shortest path map with the domain is  $O(n^4)$.
Consequently, designing a polynomial time algorithm of complexity better than $O(n^4)$ appears challenging. 
The algorithm \cite{Mitchell91} is quadratic w.r.t. these events.
Our algorithm takes first steps to provide a sub-quadratic (as a function of the number of events) solution.
The presented technique is also applicable for the case of finding an $\epsilon$-approximate shortest path over the surface of a polyhedron with weighted faces.

\begin{table}
\label{tab:timecomp}
\centering
\begin{tabular}{|p{6.0cm}|p{9.0cm}|}
\hline

\hline
Algorithm & Time Complexity \\
\hline
\hline

This paper & $O(n^6\lg(\frac{1}{\epsilon'})+n^4\lg(\frac{W}{w\epsilon'}))$ \\ 

\hline

Mitchell and Papadimitriou \cite{Mitchell91} & $O(n^8\lg(\frac{nNW}{w\epsilon}))$ \\

\hline



Aleksandrov et al. \cite{Aleks05} & $O(\frac{nK}{\sqrt{\epsilon}}\lg(\frac{NW}{w})\lg{\frac{n}{\epsilon}}\lg{\frac{1}{\epsilon}})$ \\

\hline

Cheng et al. \cite{Cheng08} & $O(\frac{\rho \lg{\rho}}{\epsilon}n^3\lg{\frac{\rho n}{\epsilon}})$, where $\rho \in [1, W] \cup \{\infty\}$ \\

\hline

Reif and Sun \cite{Reif00} & $O(\frac{nN^2}{\epsilon}\lg(\frac{NW}{w})\lg{\frac{n}{\epsilon}}\lg{\frac{1}{\epsilon}})$ \\

\hline

Aleksandrov et al. \cite{Aleks00} & $O(\frac{nN^2}{\sqrt{\epsilon}}\lg(\frac{NW}{w})\lg{\frac{1}{\epsilon}}(\frac{1}{\sqrt{\epsilon}}+\lg{n}))$ \\

\hline

\end{tabular}
\caption{Comparison of Weighted Shortest Path Algorithms}
\end{table}

Section \ref{sect:weispdefsandprops} provides definitions and properties required for the algorithm.
Section \ref{sect:algodetails} lists the details of the algorithm with pseudo-code.
The conclusions are given in Section \ref{sect:weispconclu}.
Detailed proofs and more figures are mentioned in Appendix \ref{sect:appendix}.

\section{Definitions and Properties}
\label{sect:weispdefsandprops}

Consider a planar subdivision, $\PS$, with triangular faces. 
Two faces $f, f' \in \PS$ are defined to be {\it edge-adjacent} if they share a common edge $e$.

\def\cface{{\cal F}}
\def\eseq{{\cal E}}
\begin{definition}
A sequence of edge-adjacent faces $\cface = f_1,f_2,...,f_{k+1}$ is called a {\it face-sequence} if for every $i$ in $[1,k]$, $f_i$ is edge-adjacent to $f_{i+1}$, with edge $e_i$ common to both $f_i$ and $f_{i+1}$.
And, the  sequence $\eseq (\cface)= e_1,e_2,...,e_k$  is termed the {\it edge sequence} of $\cface$.
The {\it root of the face sequence $\cface$} is the vertex $v_1$, when $f_1 $ is specified by the triangle $(v_1, v_2 , v_3)$ where $e_1= (v_2,v_3)$.
\end{definition}

\begin{definition}
Given a face $f$ and one of its edges $e$, a {\it locally $f$-free path to $x \in e$} is defined to be a geodesic path $P$ from $s$ to $x$ such that  the path does not pass through $B(x, \delta) \cap f, \forall \delta >0$, where $B(x, \delta)$ is the open ball of radius $\delta$ centered at $x$.
\end{definition}
Intuitively, a locally $f$-free path to a point $x \in e$ strikes $x$ from the exterior of face $f$ and is locally optimal.  \hfil\break

\subsection*{Background results}
\noindent The following lemmas are from \cite{Mitchell91}:
\vspace*{-.10in}
\begin{lemma}
\label{lem:crit}
Let $P$ be a geodesic path on $\PS$.  
Consider any edge $e=f \cap f'$, where $f$ and $f'$ are two adjacent faces of $\PS$.  
Let $\alpha_e=min\{\alpha_f,\alpha_{f'}\}$.  If $P$ intersects the interior of $e$ at point $q$, known as crossing point, then $p$ satisfies Snell's Law satisfying $\alpha_f \sin{\theta} = \alpha_{f'} \sin{\theta'}$, where $\theta$ and $\theta'$ are the angles made by $P$ w.r.t. the normal at $q$ in the faces $f$ and $f'$, respectively.  
If $\alpha_e=\alpha_{f'}$ and $P$ shares a segment $(a, b)$ with $e$, then $P$ satisfies local criticality condition striking $e$ from the side of $f$ at $a$ with an angle of incidence $\theta_c(f, f')=\arcsin(\alpha_{f'}/\alpha_f)$ and exits at $b$ into face $f$ at an angle $-\theta_c(f, f')$.
The segment $(a, b)$ is known as critical segment.
The closer of the two points $\{a, b\}$ to $s$ is known as critical point of entry of path $P$ from face $f$ and the other is known as critical point of exit of path $P$ into face $f$.
\end{lemma}

\begin{lemma}
\label{lem:onecritinbetween}
The general form of weighted geodesic path is a simple (not self-intersecting) piecewise-linear path which goes through an alternating sequence of vertices and edge sequences such that 
the path along any edge sequence obeys Snell's Law at each crossing point and obeys local criticality condition at each critical segment.  
Further, between any critical point of exit and the next critical point of 
entry along the path, there must be a vertex.
\end{lemma}

\begin{lemma}
\label{lem:edgeseqlen}
Let $p$ be a shortest locally $f$-$free path$.  
Let $p'$ be a sub-path of $p$ such that $p'$ goes through no vertices or critical points.  
Then, $p'$ can cross an edge $e$ at most $O(n)$ times.  
Thus, in particular, the last edge sequence of $p$ contains each edge $O(n)$ times, 
implying a bound of $O(n^2)$ on its length.
\end{lemma}

\begin{lemma}
\label{lem:critsources}
There are at most $O(n)$ critical points of entry on any given edge $e$.
\end{lemma}

\subsection*{Our Framework}

The algorithm uses rays to span the space $\PS$.
We denote a ray by the ordered pair $r= (u, v)$, where 
$u$ is the starting point of the ray
and $v$ is a  second point that defines the ray.
For a ray $r= (p, w)$, the point $p$ is termed as  $origin(r)$ and $r$ is
said to have {\it originated} from $p$.
Let  $\cRay(p)$ denote the set of rays that are originated from a point $p$.

Our algorithm simulates a discretized version of a continuous wavefront
approach. 
We  discretize the wavefront using a set of rays $\cRay (s )$ that start
from the source and fan out in a cone around the source.
When a shortest path to a vertex is determined, rays are also initialized 
from that vertex, and create the set $\cRay(v)$.
Each ray when it strikes an edge of a face of $\PS$, either refracts at the
edge or critically reflects.
As stated in Lemma~\ref{lem:crit}, the portion of the polygonal edge that 
critically reflects the rays (approximating the wavefront) 
has a {\em  critical segment} with a {\it critical point of entry}.
Consider a ray $r_c \in \cRay(v)$ that traverses over face $f$ and strikes
an edge $e = f \cap f'$ at a critical angle $\theta(f, f')$.
Note that the set of rays that strike at an angle greater than that angle
are not part of 
any shortest paths. 
A shortest path that uses $r_c$ may exit edge $e$ back onto face $f$ 
emanating from $e$ at the same critial angle $\theta(f,f')$.
We term the point, $p(e_c)$, at which $r_c$ strikes $e$ as the {\em critical
point of entry}.
We also term the entire segment of $e$, from $p$ to an endpoint of $e$
which contains points from where the shortest path comprising $r_c$ may
reflect back,
as a {\em critical edge segment}.
from
The sets comprising the vertices of $\PS$, critical points of entries, and 
critical edge segments are denoted by ${\cal V}, {\cal C}$, and ${\cal K}$
respectively.

Importantly, note that not all the rays $\cRay (s)$ defined above are 
generated during the algorithm; when a ray $r \in \cRay (s)$ is actually 
generated or  processed, then the ray $r$ is termed {\it traced}.

To achieve an $\epsilon$-approximation, 
apart from the rays that have originated from vertices and by critical
reflections from critical edge segments, 
the algorithm also introduces additional rays to fill in the gap
created by the critical reflection.
These edges atre generated  uniformly 
spaced in a cone with the critical point of entry as the apex of the cone.
For a source $v$, consider the rays from $v$ that strike an edge $e=(u,w)$
from
face $f$ at a point $p$
and let $r$ be the last ray that refracts onto face $f'$. 
Further let $r'$ be the first
ray that critically reflects and lies on $(p,w)$. 
Let $\theta$ be the angle that 
$r$ makes w.r.t normal at the point of strike on edge $e$
and $\theta'$ the angle of refraction. 
The discretization results in no refracted rays that
make an angle w.r.t. the normal at $e$ that lies 
between $\theta'$ and $90$ degrees.
To account for these, rays are generated  from $p$ that
lie in the cone defined by apex $p$ and extremal rays $r'$
and $(p,v)$. The cone lies on the face $f'$.
These additinal rays are termed {\it Steiner rays} generated from critical
points.
See Fig. \ref{fig:initcritsrc}.

\subsection*{The organization of Rays}

The set of rays originating from a vertex $v$ and  steiner rays  originating
from successive 
critical points that arise when a ray from $v$ strikes an edge at a 
critical angle, can be arranged in a tree structure ${\cal T}_R (v)$, 
known as {\em tree of rays}.
The root node of $\cT (v)$ corresponds to the vertex $v$. All the other
nodes
of the tree correspond to critical points of entry as follows:
an internal node
$w$ in the tree corresponds to a critical point of entry generated by a ray
from
$parent(w)$ striking an edge $e \in \PS$ at the critical angle.
Steiner rays are generated from this critical point as described above.
Thus with each node $u$  of the tree is associated a set of rays, 
denoted by $\cRay (u)$, generated from $u$ and  arranged in angle-sorted
order.
The tree $\cTray (v)$ thus  comprises the set of rays (traced or otherwise) 
originating from $v$ and internal nodes corresponding to the 
critical points of entry.
Note that all the rays in ${\cal T}_R (v)$ are not necessarily traced.

A path in the tree $\cTray (v)$ corresponds to a sequence of points 
$p_0(=v), p_1, p_2, \ldots, p_m$ in $\PS$, and a sequence of 
rays $r_0, r_1, \ldots, r_m$ where each point $p_i$ is a critical point and
the 
origin of a set of rays, and $r_i$ is a ray that originates from $p_i$ 
and strikes 
an edge of $\PS$ at $p_{i+1}$.
For a point $p_j$ in this sequence, the sequence of points $\{p_0, p_1,
\ldots, p_{j-1}\}$ are termed as  {\it ancestors of $p_j$ (or $r_j$)} and
points in the set $\{p_{j+1}, p_{j+2}, \ldots, p_m\}$ are termed as {\it
descendants of $p_j$ (or $r_j$)}. 

Suppose $v \in {\cal V} \bigcup \cC$ is either a vertex or a critical point 
of entry.  
Consider the rays  $\cRay (v)$ that originate from $v$ and are ordered by
the angle they make w.r.t a coordinate axis with origin at $v$.
We define two rays $r_1, r_2$ that originate from $v$ 
as {\it successive/adjacent} whenever $r_1, r_2$ are either adjacent in the 
ordered set $\cRay (v)$ or are the first and last rays in $\cRay (v)$.
See Figs. \ref{fig:vertsrc} and \ref{fig:initcritsrc}.
\begin{figure}
\centerline{\epsfysize=160pt \epsfbox{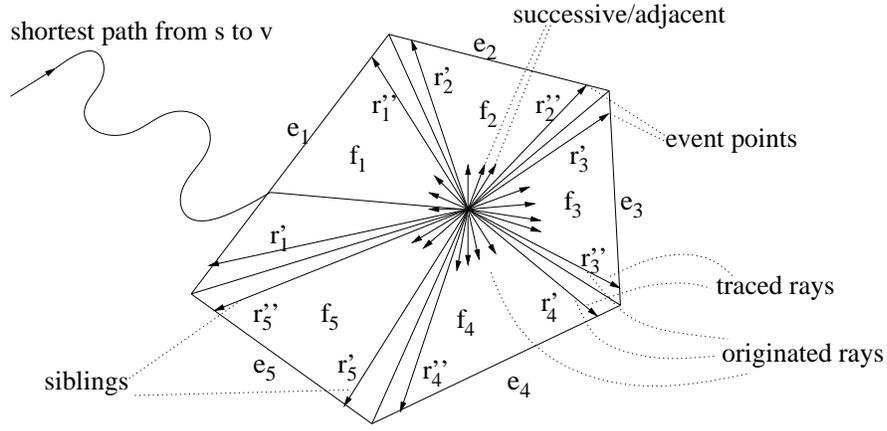}}
\caption{\label{fig:vertsrc} Tracing rays from a vertex}
\end{figure}

\begin{figure}
\centerline{\epsfxsize=320pt \epsfysize=160pt
\epsfbox{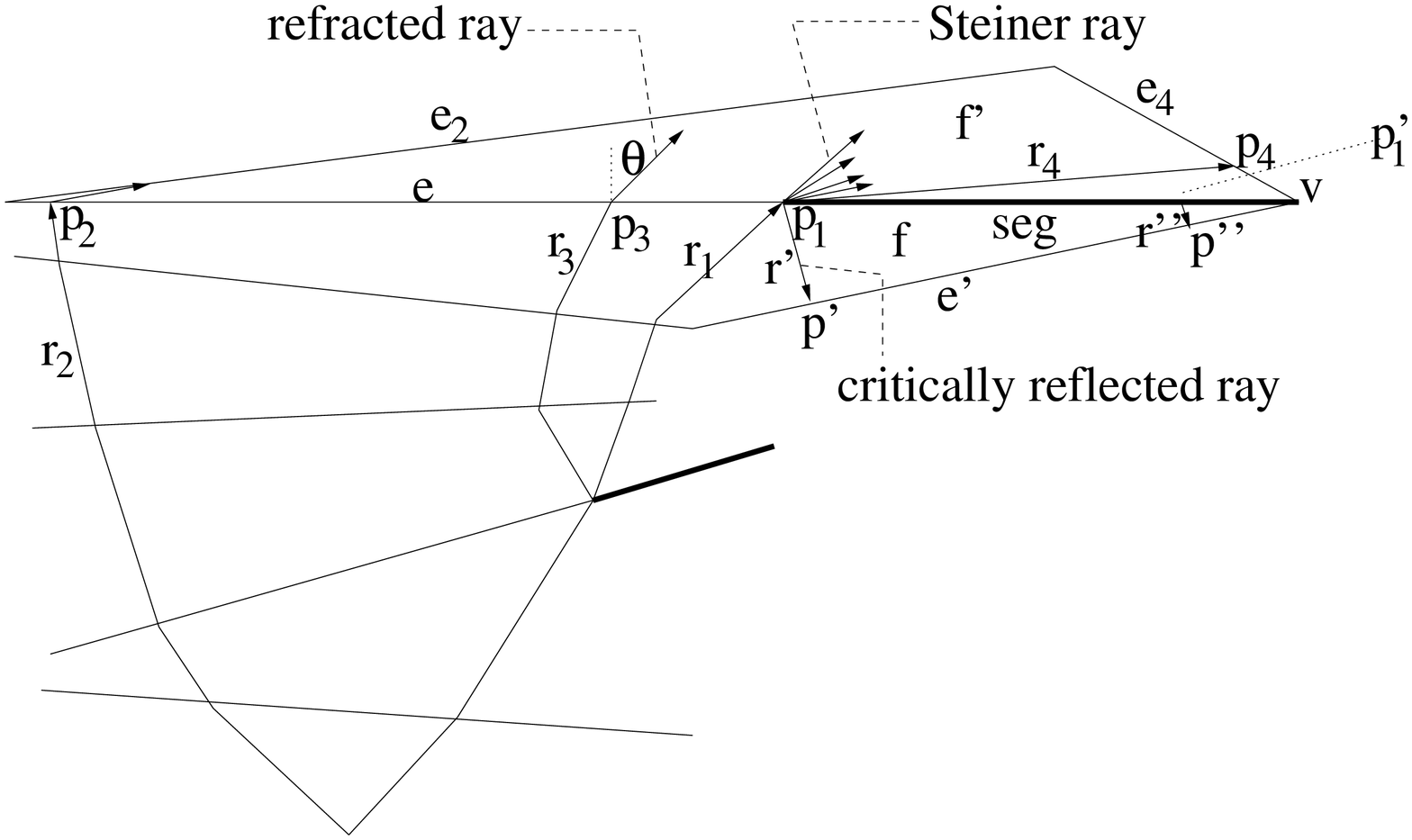}}
\caption{\label{fig:initcritsrc} Tracing critically reflected and Steiner
rays}
\end{figure}

Let $\kappa \in {\cal K}$ be a critical edge segment. 
Since the shortest path can be reflected back from
from any point on the critical edge segment, 
the algorithm will generate rays that start from various points on $\kappa$.
These rays reflect back into the face bordering $\kappa$ making an angle
with the critical segment equal to the critical angle, as specified by
Snell's Law.
Let $\cRay (\kappa)$ be the ordered set consisting of rays that originate
from $\kappa$. 
The ordering of rays in $\cRay (\kappa)$ is achieved by the sorted order of
the Euclidean distance of their origins from the critical point of entry,
$cs$, corresponding to $\kappa$.
We define two rays $r_1, r_2$ originating from $\kappa$ to be  {\it
successive/adjacent} whenever $origin(r_1)$ and $ origin(r_2)$ are adjacent
in the ordering along  $\kappa$.
See Fig. \ref{fig:reflcritsrc}. 

\begin{figure}
\centerline{\epsfysize=270pt \epsfbox{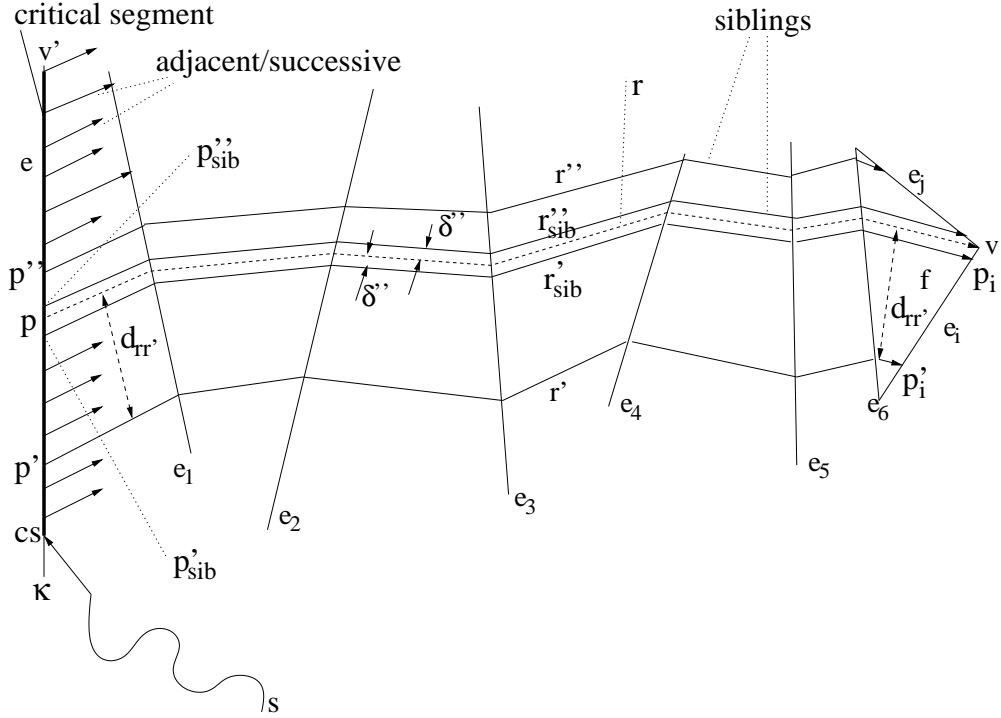}}
\caption{\label{fig:reflcritsrc} Critically reflected rays from a critical
segment}
\end{figure}

Consider a ray $r$ that has originated  from either $w \in \cTray(v)$ or
from a point $p$ for $p \in \kappa$ and $\kappa \in \cK$.
Let  $\cE (r)$ be the sequence of edges intersected by  $r$ as the 
ray propagates via a sequence of refractions.
Let $e=(v_i, v_j)$ be an edge in $\cE (r)$ and
let $\cE (r,e)$ be the initial subsequence of edges in $\cE(r)$ that
ends at $e$. 
We let $R(e,v)$ be the set of rays in $\cTray(v)$ that intersect $e$
and let 
$R(e,v,r) = \{ r'  | r' \in \cTray(v)$ and either
$\cE(r',e) $ is a subsequence of $ \cE(r,e) \}$ or the reverse is true.
Consider the set of intersection 
points of rays on $e$, 
i.e. $I(e, v,r) = \{p(e, v) \vert p(e, v)= r \cap e, \forall r \in R(e, v,
r) \}$.
Let $d(p, q)$ be the Euclidean distance between $p$ and $q$.
Define  $r'$ to be the ray such that $r' = \arg \min_{r \in R(e, v,r)} \{
d(p(r),v_i) \}$ 
i.e., $r' \in R(e, v,r)$ is the ray which intersects $e$ at the 
point $p_i$ such that among all points in $I(e, v, r)$, $p_i$ is closest to
$v_i$.
Similarly, let $r'_{sib}  = \arg \min_{r \in R(e, v, r)} \{ d(p(r),v_j) \}$.
The rays $r'$ and $r'_{sib}$ are  termed as {\it siblings} and
are computed by the relation $r'_{sib}= {\cal S} ( \cE (r , e),v)$.
These rays  will be {\em traced}.
See Figs. \ref{fig:vertsrc} and \ref{fig:reflcritsrc}.
All other rays in $R(e, v,r)$ are known as {\it intermediate rays between
sibling rays $r'$ and $r'_{sib}$}.
Note that two sibling rays can be successive rays.

Let $r$ be a Steiner ray that has originated from a critical point in a 
tree $\cTray (v)$, say at node $w \in \cTray (v)$ and $r'$ be the sibling 
of ray $r$, i.e $r'= {\cal S}(\cE(r,e),v)$, that originated at node $w'\in
\cTray (v)$.
Let $v'$ be the critical point of entry at the least common 
ancestor of $w'$ and $w$ in $\cTray (v)$ and let the ordered set $S(r)=\{r_0
, r_1, \ldots, r_k=r\}$ comprise the traced rays that originate from the
sequence of points $v', v_1, \ldots, v_k=w$, respectively, in $\cTray (v)$.
The tree path $P$ from $v'$ to $w$, 
corresponding to the sequence of critical points 
and sequence of rays in $S(r)$,
 is termed as the {\it critical ancestor path} of $r$.  
Note that the critical ancestor path of a ray $r$ is always defined with 
respect to its sibling.

We define two rays $r_1, r_2$ to be {\it parallel} if (i) $r_1$ and $r_2$
traverse the same sequence of faces ${\cal F}$ and (ii) $\forall f \in {\cal
F}$, the line segments $f \cap r_1$ and $f \cap r_2$ are parallel.

With every ray we associate a classification, $sourcetype$, from the set
\{{\it vertex, critseg, critptentry}\}.
Here, {\it vertex, critseg, critptentry} flags indicate that the ray $r$ has
originated from a vertex source, reflected from a critical segment, or
originated as a Steiner ray from a critical point of entry, respectively.

We will use $d_w(u, v)$ to be the length of the shortest weighted geodesic
path 
from $u$ to $v$.
\subsection*{Algorithm Outline}
\label{sect:weispalgo}

We outline the algorithm, which uses essentially  a discretized wavefront
approach.
The algorithm starts with a set of rays $R(s)$, that have the source $s$ as
the origin and are uniformly distributed in the ball around $s$.
As (approximate) shortest paths are discovered to vertices in $\PS$, these
vertices are used as origins to initiate rays, that are used to
determine the shortest path map w.r.t. these vertices.

Let $u$ be a vertex from which rays are originated.
Let $F$ be the set of faces incident to $u$.    
For each face $f \in F$ and for each $e \in f$, the pairs  of
sibling rays in $R(e, u)$ are traced further. 
At any point in the algorithm, there will be a number of such sibling pair
of rays.
The events corresponding to the intersection of these rays with the edges of
$\PS$ are ordered using event queues.

In general, consider the set of rays, $R(e_1, w)$.
Let $e_1 = f \cap f'$.
Suppose that one of these rays, $r$, strikes edge $e_1$ by traversing along
the 
face $f'$ and is required to be refracted and propagated 
onto the bordering face $f$.
Let $e_2$ and $e_3$ be the other two edges of $f$ such that $v = e_2 \cap
e_3$.
Propagating the sibling rays $r$ and  $r' ={\cal S}(\cE(r,e_1),v)$, 
where $r' \in \cTray (v)$,
involves computing the direction of these rays from face $f'$ to $f$, which
is accomplished using Snell's refraction formula.
Suppose one of the sibling rays strikes $e_2$ and the other strikes $e_3$. 
Then the set of rays $R(e_1,w)$ is split into two sets of rays $R(e_2,w)$
and $R(e_3,w)$ 
via a binary search procedure. 
Sibling rays in these two sets, say $r_1$ and $r_2$ are computed and an
approximation of the distance from $w$ to $v$ is computed.
Then the vertex $w$ originates a new set of rays.
We shall show that if the rays are sufficiently close enough, in angular
distance, then the approximation to the shortest distance from $w$ to $v$ is
bounded by a factor of $1+\epsilon$. 

Further, 
the rays, $R(e_1, w)$ may critically reflect at $e_1$, in which case the
critical point of entry, say $p_1$, needs to be determined. 
This generates a critical edge  segment, $\seg = (p_1, v')$ where $v'$ is an
endpoint of $e_1$.
All rays in $R(e_1, w)$ that are incident to points on the line segment
$\seg$ reflect back onto the face $f$.
The critical point of entry and the critical edge  segment each become the
origin of the following two sets of rays. 
See Fig. \ref{fig:initcritsrc}.
We let $r_1$ be the ray in $R(e_1,w)$ that strikes $e_1$ at $p_1$ and let
$r$ be the ray 
adjacent to $r_1$ incident to $e_1$ and at less than the critical angle.
\begin{enumerate}

\item
{\bf Set 1: Steiner Rays}.
Let $\theta$  be the angle which the refracted image of $r$, termed $r_R$,
 makes with the normal at edge $e$ in face $f'$.
The angular space  w.r.t. edge $e$ 
in the range $[0, \frac{\pi}{2}-\theta]$ needs to be populated with rays to
ensure an $\epsilon$ approximation. 
To achieve this, Steiner rays are generated from $p_1$ onto face $f'$ such
that

(a) There is a steiner ray, $r_{s1}$ originating  at $p_1$ onto face $f'$
that is parallel  to $r_R$.

(b) There is  a set of rays originating from $p_1$ onto face $f'$, the rays
uniformly
subdividing the angular space between $r_{s1}$ (which originates at $p_1$)
and $(p_1,v')$.

\item
{\bf Set 2: Critically Reflected  Rays}.
The other set of rays, which are critically reflected, are generated using
the points on the critical edge segment $\seg$.
These rays have the property that they are parallel and leave the edge $e$
and enter back onto face $f$ at an exit angle which is the critical angle
$\theta_c$ (Lemma \ref{lem:crit}). 
Thus the critical edge  segment originates a set of parallel rays each
separated by
a weighted distance of $\delta$, the siblings 
defining these sets of rays being easily computed.
See Fig. \ref{fig:reflcritsrc}. 
Note that, fortunately, since two critical reflections cannot occur
simultaneously on a shortest path (Lemma~\ref{lem:onecritinbetween}), a
critically reflected ray will not generate any further critical points of
entry.
In addition to the reflected rays, the entire 
critical segment from the critical source, $p_1$, to the endpoint $v'$ is
also included
as a potential shortest path to $v'$ via the critical source.

\end{enumerate}

The above description of Set 1 gives the following property.
\begin{property}
\label{prop-adj}
Two adjacent rays in $\cTray(v)$ are either parallel or originate from 
the same node $w \in \cTray(v)$
\end{property}

\subsection*{Bounding the number of traced rays}
The shortest path from the source $s$ to the destination can be split into 
sub-paths where each sub-path can be classified into types of sup-paths:

\begin{enumerate}
\item
A Type-1 path starts at a vertex in $\PS$ and only refracts along the path.

\item
A Type-2 path uses critical segments, where no two critical segments occur 
consecutively, i.e every section of the path  starts at a vertex, refracts 
and possibly critically reflects (just once) before either reaching the destination or another vertex.

\end{enumerate}
As evident from Lemma~\ref{lem:onecritinbetween}, these two paths suffice.

Our correctness proof is split into two parts.
The first part shows how to closely approximate the {\it Type-1 paths}, which do not critically reflect in-between. 
The second part shows how to approximate
{\it Type-2 paths}, which critically reflect in-between.

In order to establish that our algorithm approximates Type-1 paths, 
we first establish a bound on the number of rays to be propagated from 
the source.
\begin{lemma}
\label{thm:weisprefrnoncritical}
Let the angle between two successive rays in $\cTray(v)$, $\forall v \in {\cal V} \bigcup {\cal C}$,  be less than or equal to $\frac{w\epsilon'}{2W}(\frac{\epsilon'}{K})^{n^2}$. 
Then  a Type-1 shortest path from a  vertex $u$ to a vertex $v$ can be $\epsilon$-approximated 
using sibling rays in $\cTray(v)$.
Here $K$ is a large constant, $\epsilon' = \frac{\epsilon}{n(n^2+LW)}$, $n$ is 
the number of vertices, 
$w$ is the minimum and  $W$ is the maximum weight over all faces, 
and $L$ is the maximum length of any edge in $\PS$.
\end{lemma}
\vspace*{-.10in}
\begin{proof} 
\noindent 
Given a Type-1 path $P$ from, a vertex, say $u$, to another vertex, say $v$, we show that there exists a Type-1 path $P'$ using rays in our discretization so that $P'$ closely approximates $P$. 
This path can use Steiner rays also. 
We show that the rays originated from $u$ and successive critical points of entry, considered by the algorithm, are sufficient to discretize the space and ensure an $\epsilon$-approximation to any Type-1 path from $u$. 
We do so by (a) establishing a bound on the divergence of adjacent rays as they propagate
(b) establishing an error bound on using steiner rays to approximate rays in $\cTray(v)$
and (c) bounding the distance of a vertex from a set of rays in $\cTray(w)$, $w$ a vertex
or critical source and establishing errors in using a sequence of vertices $w_i, i=1 \ldots k$ and their corresponding $\cTray(w_i)$ to determine the shortest path. 

\paragraph{(a) Divergence of rays:}
We first consider the section of the path with no intermediate 
vertices.
We show a bound on the angle between adjacent rays so
as to limit the divergence of the rays in $\cTray(u)$.
Let two rays, say $r', r''$, originating from a vertex $u$ be separated by an angle $\gamma$, and traverse across regions with weights $\alpha_1, \alpha_2, \ldots, \alpha_p$ respectively. 
This would be the regions through which the section of the shortest path,
 under consideration, traverses.
Let the edges intersected by the rays be $e_1, e_2, \ldots, e_{p-1}$ respectively.
W.l.o.g. we consider the case when the angle between the rays $r'$ and $r''$, is increasing due to refractions at these edges and for this case to occur we assume that $W=\alpha_1 \ge \alpha_2 \ge \ldots \ge \alpha_p=w \ge 1$.
Let $\theta_1', \theta_2', \ldots, \theta_{p-1}'$ be the angles at which the ray $r'$ incidents on edges $e_1, e_2, \ldots, e_{p-1}$ respectively.
Let $\theta_2, \theta_3, \ldots, \theta_p$ be the angles at which the ray $r'$ refracts at edges $e_1, e_2, \ldots, e_{p-1}$ respectively. 
Similarly, let $\theta_1'+\delta_1, \theta_2'+\delta_2, \ldots, \theta_{p-1}'+\delta_{p-1}$ be the angles at which the ray $r''$ incidents on edges $e_1, e_2, \ldots, e_{p-1}$ respectively.
And, let $\theta_2+\delta_2, \theta_3+\delta_3, \ldots, \theta_p+\delta_p$ be the angles at which the ray $r''$ refracts at edges $e_1, e_2, \ldots, e_{p-1}$ respectively. 
Assume that $\theta_p \le \theta_{critical} \le \frac{\pi}{2}$; and, $\forall_i, \delta_i$ is small. \hfil\break
\hfil\break
For every integer $i \in [1, p-1]$, \hfil\break
\begin{eqnarray}
\alpha_i\sin{\theta_i'}=\alpha_{i+1}\sin{\theta_{i+1}} \label{eq:refr}
\end{eqnarray} 
For every integer $i \in [1, p-1]$,
\begin{eqnarray}
\lefteqn{\alpha_i\sin{(\theta_i'+\delta_i)}=\alpha_{i+1}\sin{(\theta_{i+1}+\delta_{i+1})}}  \nonumber \\
& \Rightarrow & \alpha_i\sin{\theta_i'}\cos{\delta_i}+\alpha_i\cos{\theta_i'}\sin{\delta_i} = \alpha_{i+1}\sin{\theta_{i+1}}\cos{\delta_{i+1}}+\alpha_{i+1}\cos{\theta_{i+1}}\sin{\delta_{i+1}} \nonumber
\end{eqnarray} 

Since $\delta$ is small ($-0.05 \le \delta \le 0.05$), approximating $\sin{\delta}$ with $\delta$ and $\cos{\delta}$ with one, the above equation becomes
\begin{eqnarray}
\lefteqn{\alpha_i\sin{\theta_i'}+\alpha_i\delta_i\cos{\theta_i'} \approx \alpha_{i+1}\sin{\theta_{i+1}}+\alpha_{i+1}\delta_{i+1}\cos{\theta_{i+1}}} \nonumber \\
& \Rightarrow & \alpha_i\delta_i\cos{\theta_i'} \approx \alpha_{i+1}\delta_{i+1}\cos{\theta_{i+1}}  (from (\ref{eq:refr})) \nonumber \\
& \Rightarrow & \delta_{i+1} \approx \frac{\alpha_i}{\alpha_{i+1}}\frac{\cos{\theta_i'}}{\cos{\theta_{i+1}}}\delta_i  \label{eq:delta}
\end{eqnarray}
Therefore, 
\begin{eqnarray}
\delta_p \approx \frac{\alpha_1}{\alpha_p}\frac{\cos{\theta_1'}}{\cos{\theta_2}}\frac{\cos{\theta_2'}}{\cos{\theta_3}}...\frac{\cos{\theta_{p-1}'}}{\cos{\theta_p}}\delta_1  \nonumber \\
\Rightarrow \delta_p \le \frac{\alpha_1}{\alpha_p}(\max(\frac{\cos{\theta_1'}}{\cos{\theta_2}},\frac{\cos{\theta_2'}}{\cos{\theta_3}},...,\frac{\cos{\theta_{p-1}'}}{\cos{\theta_p}}))^p\delta_1  \nonumber
\end{eqnarray}
Let $\max(\frac{\cos{\theta_1'}}{\cos{\theta_2}},\frac{\cos{\theta_2'}}{\cos{\theta_3}}, \ldots, \frac{\cos{\theta_{p-1}'}}{\cos{\theta_p}}) = \beta$. 
Then the above becomes, 
\begin{eqnarray}
\delta_p \le \frac{\alpha_1}{\alpha_p}\beta^p\delta_1 \label{eq:deltan}
\end{eqnarray}
From (\ref{eq:refr}), $\alpha_i \ge \alpha_{i+1}$, and $\theta_i' \le \theta_{i+1}$ we know that $\cos{\theta'_i} \ge \cos{\theta_{i+1}}$.
Since $\frac{\alpha_i}{\alpha_{i+1}} \ge 1$, $\frac{\cos{\theta_i'}}{\cos{\theta_{i+1}}} \ge 1$; from (\ref{eq:delta}), $\delta_{i+1} \ge \delta_i$.
Hence, $\delta_p \ge \delta_{p-1} \ge ... \ge \delta_1$. \hfil\break
\hfil\break

\begin{figure}
\centerline{\epsfxsize=400pt \epsfbox{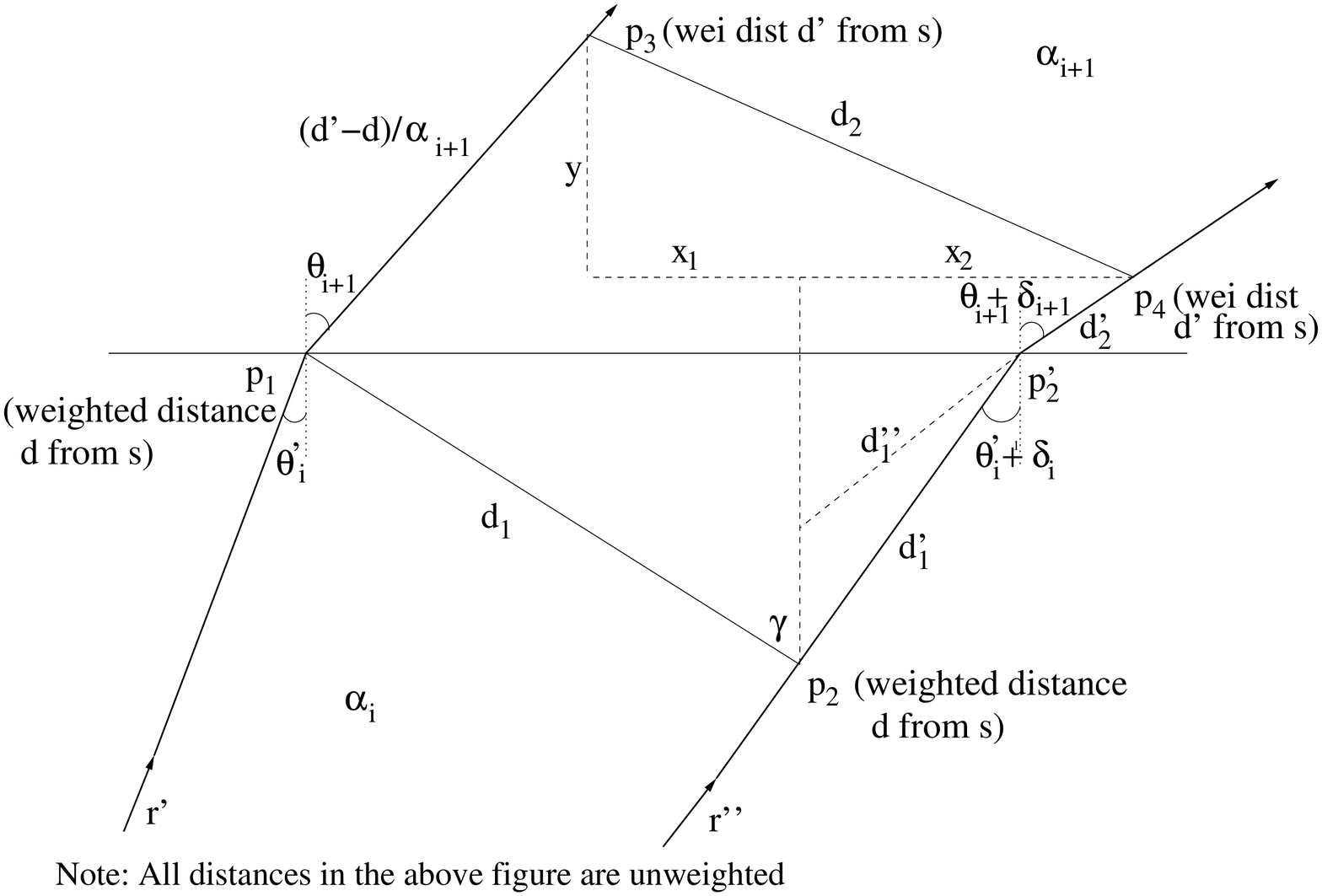}}
\caption{\label{fig:refrthm} Refraction of two successive rays from a vertex} 
\end{figure}

We next consider the distance between two points such that the
points are (i) at equal distance from a vertex $v$
(ii) lie on adjacent rays $r'$ and $r''$ in $\cTray(v)$.
Consider an edge $e_i$ on which the ray $r'$ incidents at $p_1$ with an angle of incidence $\theta_i'$ and refracts from $p_1$ with an angle of refraction $\theta_{i+1}$ whereas the ray $r''$ incidents at $p_2'$ with an angle of incidence $\theta_i'+\delta_i$ and refracts from $p_2'$ with an angle of refraction $\theta_{i+1}+\delta_{i+1}$. 
See Fig. \ref{fig:refrthm}.
Consider two points $p_1, p_2$ located on rays $r', r''$ respectively such that $p_1, p_2$ are in the region $r_i$, with weight $\alpha_i$, the 
weighted distance from $u$ to either of these points is $d$, and $\Vert p_1p_2 \Vert = d_1$. Let $r''$ be incident to $e_i$ at a larger angle than $r'$.
Also, consider two points $p_3, p_4$ located on rays $r', r''$ respectively such that $p_3, p_4$ are in the region $r_{i+1}$ with weight $\alpha_{i+1}$, 
the weighted distance from $u$ to either of these points being $d'$, and $\Vert p_3p_4 \Vert = d_2$.
We wish to establish a bound on $d_2$.
Note that it suffices to restrict attention to the above case when both
$p_3$ and $p_4$ are in $r_{i+1}$. If one of the points, say $p_4$, is in $r_i$
then the bound on $d_2$ will be smaller since $\alpha_{i+1} \leq \alpha_i$.
Also, let $\Vert p_2 p_2' \Vert = d_1'$, $\Vert p_2'p_4 \Vert = d_2'$.
We choose $d''$ such that $d_1''\sin{(\theta_{i+1}+\delta_{i+1})} = d_1'\sin{(\sin{\theta_i'}+\delta_i)}$.
\begin{eqnarray}
(d'-d) = \alpha_id_1'+\alpha_{i+1}d_2' = \frac{\alpha_id_1''\sin{(\theta_{i+1}+\delta_{i+1})}}{\sin{(\theta_i'+\delta_i)}}+\alpha_{i+1}d_2' = \frac{\alpha_i^2d_1''}{\alpha_{i+1}}+\alpha_{i+1}d_2' \label{eq:d'minusd}
\end{eqnarray}
We assume that all points at a weighted distance of $d$ satisfy  
\begin{eqnarray}
\alpha_id_1 \le d\epsilon'  \label{eq:d1alpha1}
\end{eqnarray}
for $\epsilon' < \epsilon$ and  determine conditions such that
\begin{eqnarray}
\alpha_{i+1}d_2 \le d'\epsilon' \label{eq:d2alpha}
\end{eqnarray}
\hfil\break
$\alpha_{i+1}d_2$  \hfil\break
$\le \alpha_{i+1}((x_1+x_2)+y)$ (from triangle inequality) \hfil\break
$= \alpha_{i+1}(|d_1\sin{\gamma}-\frac{(d'-d)}{\alpha_{i+1}}\sin{\theta_{i+1}}+ 
   d_1'\sin{(\theta_i'+\delta_i)}+d_2'\sin{(\theta_{i+1}+\delta_{i+1})}|+ 
   |\frac{(d'-d)}{\alpha_{i+1}}\cos{\theta_{i+1}}-d_2'\cos{(\theta_{i+1}+\delta_{i+1})}|)  \hfil\break 
= |\alpha_{i+1}d_1\sin{\gamma}-(d'-d)\sin{\theta_{i+1}} + 
  \alpha_{i+1}d_1'\sin{(\theta_i'+\delta_i)}+\alpha_{i+1}d_2'\sin{(\theta_{i+1}+\delta_{i+1})}| + 
  |(d'-d)\cos{\theta_{i+1}}-\alpha_{i+1}d_2'\cos{(\theta_{i+1}+\delta_{i+1})}|$  \hfil\break
$\le |\alpha_id_1\sin{\gamma}-(d'-d)\sin{\theta_{i+1}}+(\alpha_id_1'+\alpha_{i+1}d_2')\max\{\sin{(\theta_i'+\delta_i)},\sin{(\theta_{i+1}+\delta_{i+1})}\}|+|(d'-d)\cos{\theta_{i+1}}-\alpha_{i+1}d_2'\cos{(\theta_{i+1}+\delta_{i+1})}|$ (for $\alpha_i \ge \alpha_{i+1}$) \hfil\break
$= |\alpha_id_1\sin{\gamma}-(d'-d)\sin{\theta_{i+1}+(d'-d)\sin{(\theta_{i+1}+\delta_{i+1})}}\}|+|(d'-d)\cos{\theta_{i+1}}-\alpha_{i+1}d_2'\cos{(\theta_{i+1}+\delta_{i+1})}|$ \hfil\break
$\le |\alpha_id_1+(d'-d)(\sin{(\theta_{i+1}+\delta_{i+1})}-\sin{\theta_{i+1}})|+|(d'-d)(\cos{\theta_{i+1}}-\cos{(\theta_{i+1}+\delta_{i+1})})|$
 (from (\ref{eq:d'minusd}), (\ref{eq:d1alpha1}), and, for $0 \le (\theta_i'+\delta_i) \le (\theta_{i+1}+\delta_{i+1}) \le \pi/2, \sin{(\theta_i'+\delta_i)} \le \sin{(\theta_{i+1}+\delta_{i+1})}$) \hfil\break
$\le |d\epsilon'+(d'-d)\delta_{i+1}|+|-(d'-d)\delta_{i+1}|$  (from (\ref{eq:d1alpha1}), and, for small $\delta_{i+1}$) \hfil\break
$\le d\epsilon'+(d'-d)(2\delta_{i+1})$ \hfil\break
\hfil\break
Since $\delta_{i+1} \le \delta_p$, and, from (\ref{eq:deltan}),
$2\delta_{i+1} \le 2\delta_p \le 2\frac{\alpha_1}{\alpha_p}\beta^{p}\delta_1$
$\le 2\frac{\alpha_1}{\alpha_p}\beta^{n^2}\delta_1$ (in the worst-case, $p$ the number of 
edges intersected by a ray is $O(n^2)$)  \hfil\break
To satisfy (\ref{eq:d2alpha}), we need to have,
$2\frac{\alpha_1}{\alpha_p}\beta^{n^2}\delta_1 \le \epsilon'$ i.e., 
$\delta_1 \le \frac{1}{2}\frac{w}{W}(\frac{1}{\beta})^{n^2}\epsilon'$  \hfil\break
We thus get a bound on $\gamma$,
\begin{eqnarray}
\label{bound-angle}
\gamma =\delta_1 \le \frac{1}{2}\frac{w}{W}(\frac{1}{\beta})^{n^2}\epsilon'  \label{eq:xinbeta}
\end{eqnarray}
to achieve abound on the divergence of adjacent rays.

\paragraph{(b)Error in using steiner rays:}
The above analysis will be used to bound the approximation of the  distance
from $v$  to a  point that lies between the two adjacent rays. 
However these two adjacent rays may not exist since a ray striking 
at the critical angle reflects back into the face.
For faces $f_i, f_{i+1}$, let the edge $e_i = f_i \cap f_{i+1}$.
For a ray which incidents on $e_i$ at $\theta_i'$ from $f_i$ and refracts at $\theta_i$ along $f_{i+1}$, let $\beta = \frac{\cos{\theta_i'}}{\cos{\theta_{i+1}}}$.
Suppose the ray $r_i$ is incident on edge $e_i$ at point $p$ with an angle of incidence greater than or equal to $\theta_{ci}$.
Then the algorithm initiates a set $S$ of Steiner rays from $p$.
Among all rays in $S$, suppose the ray $r \in S$ makes the largest angle $\theta_{ri}$ with the normal to $e_i$. 
Then $\beta = \frac{\cos{\theta_{ci}}}{\cos{(\frac{\pi}{2}-\theta_{ri}})}
            \le \frac{1}{\sin{\theta_{ri}}} 
            \le \frac{K}{\epsilon'}$,
for a large constant $K$ and $\epsilon'$ as defined below.

\begin{figure}
\centerline{\epsfxsize=340pt \epsfbox{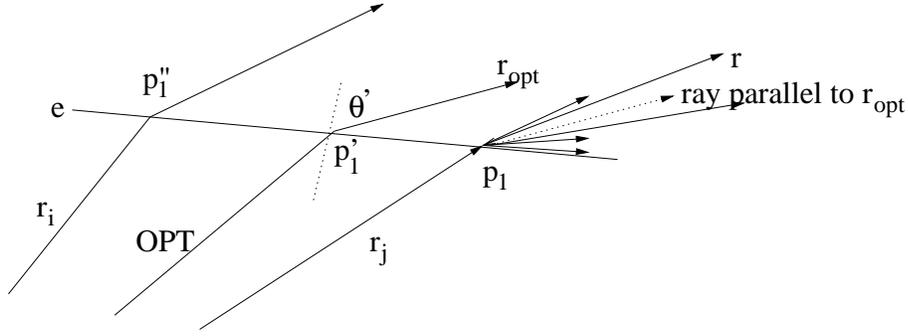}}
\caption{\label{fig:erroroptsucc} Bounding the error when an optimal path traverses between two successive rays}
\end{figure}

Consider two successive rays $r_i$ and $r_j$ from the same point, say $v_1$.
See Fig. \ref{fig:erroroptsucc}.
Let the ray $r_i$ be incident at point $p_1''$ located on edge $e$, at an angle less than the critical angle.
And, let the ray $r_j$ be incident at point $p_1$ located on edge $e$, at an angle greater than or equal to the critical angle. 
Let $d_{P}(x, y)$ represent the distance between points $x$ and $y$ along path $P$. 
Suppose an optimal shortest path $OPT$ intersects edge $e$ between $p_1''$ and $p_1$ at $p_1'$.
Let $OPT$ refracts at point $p_1'$ at an angle $\theta'$.
Further, suppose the ray $r \in S$ makes an angle $\theta''$ with the positive y-axis
and suppose that for any ray $r' \in S$, the absolute difference in the angle $r'$ makes with the positive y-axis and $\theta'$ is greater than or equal to $\vert \theta' -\theta'' \vert$. 
Consider a path $P$ along the ray $r_j$ from $v_1$ to $p_1$ and along the ray $r$.
Then the recurrence $d_{P}(v_1, v_2) = d_{P}(v_1, p_1)+d_{P}(p_1, v_2)$ represents the distance from point $v_1$ to point $v_2$ along path $P$.
Let $p_1''$ be at weighted distance $d_1''$ from $v_1$ and, let $p_1$ 
be at weighted distance $d_1$ from $v_1$.
For $d' = \max(d_1, d_1'')$ and $\alpha_e$ being the weight of edge $e$ and with $\alpha_id_1 \le d\epsilon'$, 
we have $\alpha_e\Vert p_1p_1'' \Vert \le d'\epsilon'$ 
(by the choice of $\gamma$ in equation~(\ref{bound-angle})).
Then $d_{P}(v_1, v_2) \le (1+\epsilon')(d_{OPT}(v_1, p_1')+d_P(p_1, v_2))$.
Let $d$ be the optimal shortest distance from $v_1$ to $v_2$.
Since there can be at most $O(n^2)$ critical points of entry along any path, 
expanding the recurrence yields $d_{P}(v_1, v_2) \le (1+\epsilon')^{c_1n^2}d$.
Let the error in computing the shortest distance from $v_1$ to $v_2$ be denoted by ${\cal Y}$. 
Then for small $\epsilon'$, 
\begin{eqnarray}
{\cal Y} \le dc_1\epsilon'n^2
\label{eq:12}
\end{eqnarray}
\begin{figure}
\centerline{\epsfysize=220pt \epsfbox{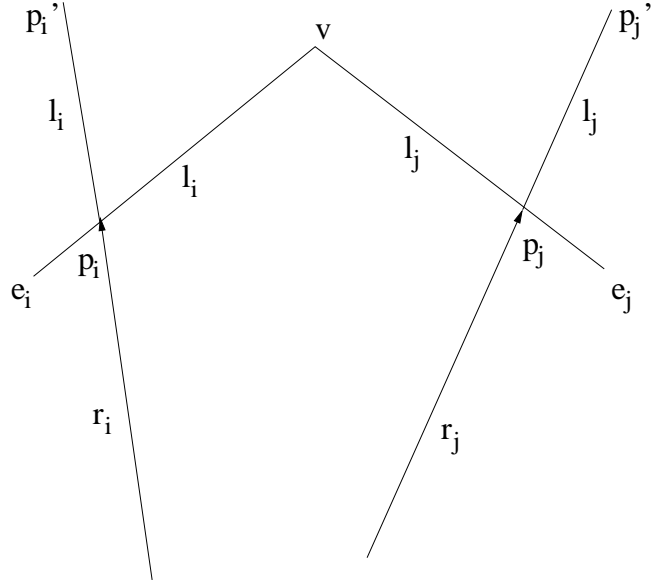}}
\caption{\label{fig:errorvertsucc} Bounding the error when a vertex falls in between two successive rays}
\end{figure}

\paragraph{(c) Error in reaching the destination:}.
We finally show that the distance from an origin $u$ to another vertex $v$
can be approximated  by the sibling rays in $\cTray(u)$.
We let ${\cal E}$ be the sequence of edges traversed by the optimum
path from $u$ to $v$.
Let the vertex $v = e_i \cap e_j$.
Consider the sibling rays in the sequence of rays that intersect 
the sequence of edges ${\cal E}$.
Let the  last edge in the sequence be $e_k$ before the shortest path 
reaches $v$ after which the shortest path enters the triangle $e_i,e_j,e_k$
where $v$ is common to $e_j$ and $e_k$.
We will compare the distance between the ray representing the optimal
path from $u$ to $v$  and the rays in $\cTray(u)$. More specifically
we will consider the rays in $\cTray(u)$ that intersect the sequence
${\cal E}$.
We consider the two sibling rays $r_i \in \cTray(u)$ and $r_j \in \cTray(u)$ 
such that the ray $r_i$ is incident to edge $e_i$ at $p_i$, closest to $v$ and 
the ray $r_j$ is incident to edge $e_j$ at $p_j$, again closest to $v$. 
At least one of these sibling rays exist from the set of rays that
intersect ${\cal E}$.
Note that the rays in $\cTray(u)$ that intersect the sequence
${\cal E}$ split and intersect either both $e_i$ and $e_j$ or 
at least one of the edges $e_j$ or $e_k$.
See Fig. \ref{fig:errorvertsucc}.
We consider the first case. The second case is similar.
Suppose $\Vert vp_i \Vert = l_i, \Vert vp_j \Vert = l_j$.
Let $p_i$ be at weighted distance $d_i$ from $u$ and $p_j$ be at 
weighted distance $d_j$ from $u$.
For $d' = \max(d_i, d_j)$, from the bound (\ref{bound-angle})
we know that $\Vert p_ip_j \Vert \le d'\epsilon'$.
Let $p_i'$ be located on ray $r_i$ extended such that $\Vert p_ip_i' \Vert = l_i$.
Similarly, let $p_j'$ be located on ray $r_j$ 
extended such that $\Vert p_jp_j' \Vert = l_j$.
From above, we know that $\Vert p_i'p_j' \Vert \le \max((d_i+\alpha_{e_i}l_i)\epsilon', (d_j+\alpha_{e_j}l_j)\epsilon')$ where $\alpha_{e_i}$ is the weight along the edge $e_i$ and $\alpha_{e_j}$ is the weight along the edge $e_j$. 
Let $L$ be the maximum edge length and, let $W$ be the maximum weight.
Then $\Vert p_i'p_j' \Vert \le \max(d_i\epsilon', d_j\epsilon')+LW\epsilon'$.
Then at the vertex $v$, the error, ${\cal X}$,
i.e. the distance from the vertex to a point on a traced ray, is upper bounded as follows: 
\begin{eqnarray}
{\cal X} \le LW\epsilon' 
\label{eq:13}
\end{eqnarray}
Even in the case when the set of rays do not split
and there is only one sibling ray closest to $v$ to consider, 
the vertex $v$ on the
ray corresponding to the shortest path from $u$ to $v$ is at 
distance less than $d_w(u,v)\epsilon'$ from the sibling ray 
(from \ref{bound-angle}). The rest of the proof is similar to the first case.

From (\ref{eq:12}) and (\ref{eq:13}), and, given that there are $c_2n$ vertices on a
shortest path where successive errors can acculumuate, 
the total error in computing the distance from $u$ to $v$ is at most $dcn\epsilon'(n^2+LW)$, for $c=\max(c_1, c_2)$.
Since this should not exceed $d\epsilon$, we choose $\epsilon' = \frac{\epsilon}{n(n^2+LW)}$.
Note that the approximation is achieved by rays which are closest to the vertex $v$.
These rays are sibling rays in $\cTray (u)$.
\end{proof}

\def\cseg{{\cal CS}}
\def\kappa1{{\cal CS}}
\begin{lemma}
\label{thm:weisprefrcritical}
Let $P$ be a Type-2  shortest path from a vertex $v$ to another 
vertex $w$ on $\PS$ with at least one critical segment ${\cal CS}$ in-between.
If the angle between successive rays is as specified in Lemma~\ref{thm:weisprefrnoncritical}
then an  $\epsilon$-approximate Type-2 shortest path can be found from
sibling rays  in $\cRay (v), \forall v \in {\cal V} \bigcup {\cal C}$.
\end{lemma}
\begin{proof}
The optimum path $P$ can be partitioned into sub-paths, each sub-path going from
a vertex $v$ to a vertex $w$. Let this sub-path use a critical segment and
be partitioned as: 
a path  $P_1$ from $v$ to the critical segment, $\cseg \in \cK$, 
and a path $P_2$ from $\cseg$ to $w$.
Let $e$ be the edge on a face $f$ such that the critical 
segment $\cseg$ lies  on $e$ and reflects rays back onto face $f$. 
Also, let the critical segment $\cseg$ have a critical point of entry $cs$.
We first show that a good approximation to the path from
$\cseg$ to $w$ can be found. 
If $w$ is the endpoint of $\cseg$ then we are done since the distance to
the endpoint from the critical source is included in consideration.
Otherwise consider rays  generated from  $\cseg$ that are parallel
and separated by a weighted distance of $\delta$.
As in the previous lemma, 
we consider in detail the case when  $w$ lies in-between 
two adjacent (parallel) rays from $\cseg$, $r_1$ and
$r_2$, which are sibling rays. 
As in lemma~\ref{thm:weisprefrnoncritical}, the error in 
computing the distance from $w$ to a
point on one of the rays is bounded by $LW\epsilon'$. 
The optimal path  $P_2$ has a first segment that 
is parallel to and lies in between a pair of rays $r'$ and $r''$. The
source of the optimum path lies in between 
the origin of  the two rays, $origin(r')$ and $origin(r'')$, that lie on $\cseg$.
Since  the weighted distance between $origin(r')$ and $origin(r'')$ is at most $\delta$,
the error in determining a path along one of these 
rays is $\delta +LW\epsilon'$.
Further the error in determing the shortest path from $v$ to $cs$ is
at most ${\cal Y} \le dc_1\epsilon'n^2$ (equation.~\ref{eq:12}).
Combining the above errors, choosing $\delta \leq \epsilon'$ and noting
that there are at most $O(n)$ vertices on the shortest path gives the
required error bound as in lemma~\ref{thm:weisprefrnoncritical}.

\end{proof}

\section{Algorithm Details}
\label{sect:algodetails}

In the algorithm, the main procedure is the event handler.
Event handling is categorized into the following types:
\begin{enumerate}
\item Shooting rays from a vertex
\item Refracting a ray
\item Finding the critical point of entry
\item Shooting critically reflected rays from a critical segment and  shooting Steiner rays from a critical point of entry
\end{enumerate}
The description of procedures that process the events is described below.
The  determination and handling of these events is further detailed with the pseudo-code.\\ \\

\subsection{Main Procedure}

The main procedure is {\it weightedEuclideanSP}.
It implements a version of continuous Dijkstra's algorithm as follows: An event heap is maintained, the heap comprising points  ordered by their distance to the source vertex.
The event may be a vertex, a point of strike of a ray, $r$, on edge $e_i$ or a critical source.
Depending on whether the ray strikes the edge at an angle less than or greater than the critical angle, the ray is refracted or dropped from further consideration respectively.
Let a ray and its sibling be such that one of them strikes $e_i$ at an angle 
greater than the critical angle and the other at an angle 
less than the critical angle. 
The point at which a ray, starting from either of the sources of these rays, 
strikes the edge $e_i$  at the critical angle is determined. 
The distance to this point is determined and added to the heap.
This point is termed as the {\em critical source}.
Let $e$ be an edge whose one endpoint is $t$.
The algorithm terminates whenever the event is a vertex event with the vertex being $t$.

The actions required at each event are listed below:
\begin{itemize}

\item[(i)] 
A vertex event triggers the procedure {\it initiateVertexSource} which initiates rays from a given vertex $v$.  
\item[(ii)]
An event which represents a critical point of entry $p_1$ uses the procedure {\it initiateCriticalSource} to initiate both critically reflected rays from the corresponding critical segment and Steiner rays from $p_1$ respectively.
From each critical segment, initially only two critically reflected rays are progressed as siblings.

\item[(iii)]
An event where a ray strikes an edge is handled in two parts.  
The procedure {\it refractRay} extends the ray across the current face $f$.
It determines whether the given ray $r'$ and its current sibling are incident onto the same edge of $f$.  
If not, it finds the correct siblings for both of these rays. 
Then it adds events corresponding to the rays striking the edges of the 
face $f$.
The procedure {\it findCriticalPointOfEntry} finds the ray which 
originates from a vertex/critical point of entry and is incident to an 
edge at a critical angle. 

\end{itemize}

\begin{algorithm}
\caption{weightedEuclideanSP}
\begin{algorithmic}[1]
\vspace*{.1in}

\STATE Initialize  distance between all pairs of 
vertices, $d_w(v, w), \  v, w \in {\cal P}$ to be $\infty$.
\STATE Push source $s$ to the event queue

\WHILE{the event queue is not empty}

	\STATE $evt \leftarrow$ pop the event queue 

	\STATE \COMMENT {The following cases are both mutually exclusive and exhaustive}
	\IF{$evt$ represents a vertex $v \neq t$}
		\STATE initiateVertexSource($v$)
	\ENDIF 

	\IF{$evt$ represents a critical point of entry $cs$}
		\STATE initiateCriticalSource($cs$)
	\ENDIF

	\IF{$evt$ represents a ray $r'$ which incidents to edge $e_i = f_i \cap f_{i+1}$ from face $f_i$ at an angle less than $\arcsin(\frac{\alpha_{f_{i+1}}}{\alpha_{f_i}})$}
		\STATE refractRay($r', e_i$)
	\ENDIF
	\IF{$evt$ represents a ray $r'$ with $r'.sibling=r''$ such that $r'$ incidents to edge $e_i=f_i \cap f_{i+1}$ from face $f_i$ at an angle $\le \theta_{ci}$ and $r''$ incidents on the edge $e_i$ from face $f_i$ at an angle $\ge \theta_{ci}$, for $\theta_{ci}=\arcsin(\frac{\alpha_{f_{i+1}}}{\alpha_{f_i}})$ (note that here $r'$ and $r''$ can be replaced one for the other)}
         
           \IF{ source of ray is not on a critical segment}
		\STATE findCriticalPointofEntry($r', r'', e_i$)
          \ENDIF
	\ENDIF
\IF{$evt$ corresponds to a vertex event where the vertex is $t$}
   \STATE empty event queue
\ENDIF

\ENDWHILE
\end{algorithmic}
\end{algorithm}
To improve the efficiency, as outlined before, we only choose to maintain 
sibling rays, rather than tracing all the rays that have originated 
from a point.
Since an intermediate ray $r$ between two sibling rays always has the same 
intersected  edge sequence ${\cal E} (r)$ as one of these 
sibling rays, the intermediate ray is not required to be traced until 
it is necessary.
These binary searches are detailed in the following procedures.\\ 
\subsection{Procedure initiateVertexSource}
This procedure initiates rays from a vertex $v$, including the source. 
Every face  that includes the vertex $v$ is considered.
For each such face $f$, and the  edge $e_i=(u, w) \in f$ that is not incident 
to $v$, two rays from $R(e_i,v)$, $r_i'$ and $r_i''$ are  determined 
such that $r_i$ is closest to $u$ and $r_i'$ closest to $w$ 
(in fact the procedure discovers successive rays in $\cTray(v)$ that 
are closest to the endpoints of the edge $e_i$).
These rays $r_i$ and $r_i'$ are the sibling rays which will be propagated 
further.
The distance at which these rays will intersect $e_i$ is 
determined and added to the heap.

A binary search on $\cTray(v)$ is used to ensure that rays in $\cTray(v)$ are 
correctly generated along with sibling rays.
\begin{algorithm}
\caption{initiateVertexSource($v$)}
\begin{algorithmic}[1]
\REQUIRE{
A vertex $v$.  
Let $F$ be the set comprising all the faces incident to $v$. 
Let the set of edges defining faces in $F$ be $E'$.
Also, let $E'' (\subseteq E')$ be the set of edges which incident to $v$.
See Fig. \ref{fig:vertsrc}.
}
\ENSURE{Initiates rays from the vertex $v$}
\vspace*{.1in}
\FOR{Each vertex $v' \ne v$ such that $v'$ incident to a face $f \in F$}
	\STATE Binary search over the rays originating from $v$ to determine two successive rays, $r_i', r_j''$, such that for $e_i, e_j \in E'-E'', v' = e_i \cap e_j$, where $r'_i = \arg \min_{r \in R(e_i,v) } \{ d(p(r),v') \}$ and $r'_j = \arg \min_{r \in R(e_j,v) } \{ d(p(r),v') \}$.
        The rays $r_i', r_j'$ are the ones traced further.
\ENDFOR
\FOR{Each edge $e_i$ in $E'-E''$}
	\STATE Let rays $r_i', r_i''$ be sibling rays in $R(e_i,v)$.
	Let the point $p_i' = r_i' \cap e_i$ and, let  $p_i''= r_i'' \cap e_i$.
	Also, let $f$ be the face on which the line segments $vp_i'$ and $vp_i''$ resides.
	Push new events to the heap representing the intersection of ray $r_i'$ with $e_i$, 
which occurs at a distance $d_{i'}=\alpha_f \Vert vp_i' \Vert$ from $v$ 
and the intersection of ray $r_i''$ with $e_i$, which occurs at a 
distance $d_{i''} = \alpha_f \Vert vp_i'' \Vert$ from $v$.
\STATE
Let $v'$ and $v''$ be the endpoints of $e_i$ closest to  $p_i'$ and $p_i''$, 
respectively. Update distance from source $v$ to $v'$ in heap by 
$ \min \{d_w(v, v'), d_{i'}+ d(p_i', v')\}$
and to $v''$ by $ \min \{ d_w(v, v''), d_{i''}+d(p_{i''},v'') \}$.
 
\ENDFOR
\end{algorithmic}
\end{algorithm}
The following is clear from the above discussion:
\begin{lemma}
\label{lemma:initiateVertexSource}
Procedure {\it initiateVertexSource} correctly initiates the tree $\cTray(v)$.
\end{lemma}

\subsection{Procedure initiateCriticalSource}
When a critical point $p$ on edge $e = f \cap f'$ of face $f$ is encountered, there are two classes of rays that need to be propagated.
One is the set of rays that are to be critically reflected back onto $f$.
And the other is the  set of Steiner rays over the face $f'$ that are required to be generated from the critical source $p$. 
As required in the proof of Lemma~\ref{thm:weisprefrnoncritical} these 
Steiner rays are required to achieve the desired approximation as they serve to fill in for the set of rays on face $f'$ that strike the edge $e$ at an angle just less than the critical angle.
More specifically, let $r_1$, with $r_2$ as the designated sibling, be the ray 
that strikes $e$ at point $p$. 
Let this ray $r_1$ be an element of  $\cTray(w)$ and let $u$ be the origin of $r_1$.
Note that $u$ may be $w$ itself or a critical source.
Consider the ray $r_3$ which is adjacent to $r_1$ in $\cTray(w)$ and which would be refracted at edge $e$.
Simply refracting $r_3$ would not be sufficient since the rays incident to $e$ at angles between the angle $\alpha$ at which  $r_3$ strikes $e$ and the critical angle, refract to cover a wide range of angles at the adjacent face.
These rays are replaced by  rays that originate from $p$.
A node corresponding to the critical source $p$ is created as a child of $u$ in $\cTray(w)$ and a set of rays $\cRay(p)$ are initiated.
Sibling rays that are propagated forward are computed as follows:
One sibling ray, termed $r_4$, is the ray almost parallel to $e$, in fact which makes an angle $\epsilon'/K$
w.r.t. $e$ and the other is the sibling of the ray $r_1$, termed $r_2$, which struck $e$.

These rays are propagated across the face $f'$. Let $e_2$ and $e_4$ be the edges of $f'$ that $r_2 $ and $r_4$  strike.
If $e_2$ is the same as $e_4$ then the rays $r_2$ and $r_4$ are siblings.
Otherwise, new siblings are created using the procedure $findSplitRays$. 
Let $v$ be the vertex common to $e_2$ and $e_4$.
The shortest distance to $v$ from amongst the rays in $\cTray(w)$ lying between $r_2$ and $r_4$ is computed and the event heap corresponding to the distance from the source to the vertex source $v$ is updated. Note that when the shortest distance to $v$ will be discovered as the minimum distance event in the heap, $v$ will be a vertex source.

\begin{algorithm}
\caption{initiateCriticalSource($p_1$)}
\begin{algorithmic}[1]
\REQUIRE{
Critical point of entry $p_1$ on edge $e= (u, v)$ and 
$r_1 \in \cTray(w)$ the ray which caused the critical point of entry $p_1$.
Critical segment $seg = (p_1, v)$.
Let $f$ be the face onto which the critically reflected rays from $seg$ traverse.
Let $\theta_c$ be the critical angle of incidence of edge $e$.
Also, let $r_1.sibling = r_2$, and, let $p_2 = r_2 \cap e$.
See Fig. \ref{fig:initcritsrc}.
}
\ENSURE{Initiates critically reflected and Steiner rays from $seg$ and $p_1$ respectively}
\vspace*{.1in}
\STATE 
Let $p_1'$ be a point at $\Vert p_1v \Vert - \delta'$ distance from $p_1$, for $\delta' < \epsilon'$.
Let $r', r''$ be critically reflected (parallel) rays, making an angle $-\theta_c$ with the normal to $e$, and originate from points $p_1$ and $p_1'$ respectively.
Let $p'$ be the point at which the ray $r'$ incidents on an edge $e'$ of $f$ and, let $p''$ be the point at which the ray $r''$ incidents on an edge $e''$ of $f$. 
Note that $e'$ is not necessarily distinct from $e''$.
Set $r', r''$ as siblings.
Push new events $evt', evt''$ to the heap which represent the intersection 
of $r', r''$ with the edges of $f$ where $evt'$ occurs at a distance $\alpha_f \Vert p_1p' \Vert$ from $p_1$ and $evt''$ occurs at a distance $\alpha_e \Vert p_1p_1' \Vert + \alpha_{f} \Vert p_1'p'' \Vert$ from $p_1$. 
Also, push an event to the heap which corresponds to the 
critically reflected ray reaching
vertex $v$. 
This event  occurs at a distance $\alpha_e \Vert p_1v \Vert$ from $p_1$. 
\IF{$e' \neq e''$}
\STATE 
findCriticallyReflectedray$(r',r'',v')$ where $v'$ is the 
vertex common to $e'$ and $e''$  
\STATE 
Let $r_3$ be the ray adjacent to $r_1$ in $\cTray(w)$ which refracts when it incidents via $f$ to a point $p_3$ located on line segment $p_1p_2$.
Suppose $r_3$ refracts from $p_3$ at an angle $\theta$ onto face $f'$.
A set $S$ of Steiner rays, which make angles in the range $[\frac{\epsilon'}{K}, \frac{\pi}{2}-\theta]$ with edge $e$ and traverse over $f'$, are added to $\cTray(w)$, at a new node $p_1$, i.e. the rays originate from $p_1$.
Here, $K$ is a large constant.
Note that if  $\frac{\pi}{2}-\theta < \frac{\epsilon'}{K}$ then
the range $[\frac{\epsilon'}{K}, \frac{\pi}{2}-\theta]$ is $\emptyset$, and no rays originate from $p_1$.
Let $r_4 \in S$ be the ray which makes an angle $\frac{\epsilon'}{K}$ with edge $e$.
Suppose the ray $r_4$ first intersects with an edge $e_4$ of $f'$ 
at point $p_4$.
Set the rays $r_2$ and $r_4$ as siblings.
Push new event to the heap which represents the ray $r_4$ traversing to $p_4$ and this event occurs at distance $\alpha_{f'} \Vert p_1p_4 \Vert$ from $p_1$.
Trace $r_2$ over face $f'$ by pushing an event corresponding to the intersection of $r_2$ with an edge $e_2 (\neq e)$ of $f'$, termed $p_2'$.
Note that $e_2$ is not necessarily distinct from $e_4$.
\IF{ $e_2  = e_4$}
\STATE Update shortest distance to endpoints of $e_2$ from  $p_2'$ and $p_4$.
\ELSE
\STATE findSplitRays() and update shortest distance to the endpoints of $e_2$ and $e_4$ from the siblings.
\ENDIF
\ENDIF

\end{algorithmic}
\end{algorithm}

\begin{lemma}
\label{lemma:initiateCriticalSource}
Procedure {\it initiateCriticalSource} correctly computes a new node $w$ corresponding to  a critical source on an edge $e$ and the set of rays ${\cal R}(w)$.
It also correctly computes sibling rays incident to edges other than $e$ of $f'$.
\end{lemma}

\subsection{Procedure refractRay}

In procedure {\it refractRay}, a subset  of rays in $\cTray(s)$ that strike an edge $e_i$ 
before traversing a face $f$ are refracted.
These rays lie in between $r'$ and its sibling $r''$ that strike $e_i$.
Let $e_i, e_j, e_k$ be the edges of $f$.
And let $v$ be common point to $e_j$ and $e_k$.

Suppose $r$ and $r'$, when refracted, strike different edges of the 
face, say $e_j$ and $e_k$, respectively.
If $r$ and $r'$ have originated from a critical segment source, ${\cal CS}$ 
then a ray $r_c$ originating from ${\cal CS}$ and that strikes $v$ is computed. 
Also are determined the siblings of  $r$ and $r'$, which would be the rays adjacent to $r_c$ in $\cTray({\cal CS})$.

Alternatively, $r$ and $r'$ originated from a vertex source $w$.
Then the ray in $\cTray(w)$ that strikes $v$ is computed and the siblings of $r'$ and $r''$,  $r_1$ and $r_2$ respectively, are determined. 
This is done using procedure {\it findSplit Rays}.
The procedure {\it findSplitRays} proceeds as follows:
The critical ancestor paths, $P'$ and $P''$ of  $r'$ and $r''$ are determined.
The siblings must lie in the set of rays generated from vertices along the paths
$P'$ and $P''$. These set of rays are searched via binary search.
The events at which these rays strike $e_i$ 
and $e_j$, respectively, are computed.

\begin{lemma}
\label{lemma:refractRay}
Procedure {\it refractRay} correctly processes refracted rays across a face $f$ and updates sibling rays incident onto edges  of the face $f$. 
\end{lemma}
\begin{algorithm}
\caption{refractRay($r', e_i$)}
\begin{algorithmic}[1]
\REQUIRE{The ray $r'$ which incidents on edge $e_i$ at an angle less than $\theta_{ci}$.  
Let $r'.sibling = r''$ on the same edge $e_i$.}
\ENSURE{The rays $r'$ and $r'.sibling$ are traced further.}
\vspace*{.1in}
\STATE{Let $f$ comprising edges $e_i$, $e_j$, $e_k$ be the face over which the rays $r'$ and $r''$ are to be traced}
\IF{ $r'$ and $r''$ traced forward are incident onto different edges $e_j$ and $e_k$} 
	\IF{$r'$ and $r''$ are originated from the same critical segment}
		\STATE findCriticallyReflectedRay($r', r'', v$) where $v=e_j \cap e_k$.
	\ELSE
		\STATE findSplitRays($r', r'', e_i, e_j$)
	\ENDIF
\ELSE
\STATE Push events corresponding to rays $r'$ and $r''$ striking edges $e_j$ and $e_k$. 
\ENDIF
\end{algorithmic}
\end{algorithm}

\begin{algorithm}
\caption{findSplitRays($r', r'', e_i, e_j$)}
\begin{algorithmic}[1]
\REQUIRE{
For the rays $r'$ and $r'.sibling=r''$, such that $r',r'' \in \cTray(w)$, let $v'$ be the least common ancestor of $r'$ and $r''$.
Let $P'$ be the critical ancestor path of ray $r'$ with $v'$ as the origin.
Let $r'$ be incident to edge $e_i$ at point $p'$ and $r''$  be incident to edge $e_j$ at point $p''$, for $v = e_i \cap e_j$.
Similarly, let $P''$ be the critical ancestor path of ray $r''$ with $v'$ as the origin.
Let $REG$ be the region bounded by the rays in the critical ancestor path of $r'$ and the critical ancestor path of $r''$, 
rays  $r'$, $r''$, and the segment $p'p''$.
Let $R_1$ be the set of rays such that a ray $r \in R_1$ if and only if $r$ originates from a critical point of entry/vertex located on the critical ancestor path $P'$ and ray $r$ lies in the open region $REG$.
Similarly let $R_2$ be the set of rays such that a ray $r \in R_2$ if and only if $r$ originates from a critical point of entry/vertex located on the critical ancestor path $P''$ and ray $r$ lies in the open region $REG$.
See Fig. \ref{fig:binsrchvrtsrc}(a).
}
\ENSURE{
Finds two rays $r_1, r_2$ which intersect edges $e_i, e_j$ respectively such that among all the rays in $R_1$ and $R_2$, the rays $r_1, r_2$ are the closest rays to vertex $v$ along the edges $e_i$ and $e_j$ respectively.} 
\vspace*{.1in}
\STATE Do binary search of the rays in $R_1$ and the rays in $R_2$ to find two rays $r_1$ and $r_2$ where $r_1$ intersects $e_i$ and $r_2$ intersects $e_j$ such that $r_1$ and $r_2$ are either successive rays originating from the same origin or adjacent origins on a critical ancestor path.
Set $r_1, r'$ and $r'', r_2$ as siblings.
Push two events to the heap where one represents the intersection of ray $r_1$ with  $e_i$ and the other represents the interaction of  ray $r_2$ with  $e_j$.
\STATE Let rays $r_1, r_2$ intersect edges $e_i, e_j$ respectively at points $p_1, p_2$.
Then $\min(d_w(s,p_1)+ \alpha_ed(p_1v), d_w(s,p_2)+ \alpha_e d(p_2,v)$, 
is the event required to be update the shortest path to $v$ in the event heap.
\end{algorithmic}
\end{algorithm}

\subsection{Procedure findCriticallyReflectedRay}
\def\seg{{\cal CS}}
This procedure is invoked in the previous procedure when the ray to be refracted arises from a critical segment $\seg$.
Consider when two critically reflected sibling rays $r', r''$ from the same critical segment $\seg$ are propagated across a face and intersect  two distinct edges $e_i, e_j$ respectively. 
This procedure finds new rays $r'_{sib}$ and $r''_{sib}$ which are siblings to $r', r''$ respectively so that $r'_{sib}$ and $r''_{sib}$ originate from $\seg$.

\begin{algorithm}
\caption{findCriticallyReflectedRay($r', r'', v$)}
\begin{algorithmic}[1]
\label{algo:findtwoclosest}
\REQUIRE{Two critically reflected rays $r', r''$ originating from the same critical segment, say $\seg$, with $r'.sibling=r''$; $r'$ incidents on edge $e_i$ and $r''$ incidents on edge $e_j$ for $v = e_i \cap e_j$.
See Fig. \ref{fig:reflcritsrc}.}
\ENSURE{Finds a critically reflected ray $r$ originated from a point on $\seg$ such that $r$ incidents to $v$.  
Also, finds new siblings $r'_{sib}$ and $r''_{sib}$ for $r'$ and $r''$.}
\vspace*{.1in}
\STATE 
Let $p'=origin(r')$ and $p''=origin(r'')$.
Let $l$ be the length of the edge sequence intersected by $r'$ i.e., $l=|\cE(r)|$.
Consider a critically reflected ray $r$ such that $r=origin(p)$ on the line segment $p'p''$ which passes through $v$, is parallel to $r'$ and  $\cE (r)$ is the same as the first $l-1$ elements of $\cE(r')$.
Update event in heap corresponding to $d_w(\seg, v)$, the weighted distance
from $\seg$ to $v$.
\STATE 
Construct the ray $r'_{sib}$, $r'_{sib}=origin(p'_{sib})$ where $p'_{sib}$ is
located on the line segment $p'p''$ such that \\
(i) $r'_{sib}$ is parallel to $r'$\\ 
(ii) $\cE (r')$ is the same as $\cE (r'_{sib})$ \\
(iii) $d(r,r'_{sib}) = \delta$ (for $\delta < \epsilon'$).\\
And construct the ray $r''_{sib}$ be the mirror image of $r'_{sib}$ w.r.t. $r$, i.e. $r''_{sib}=origin(p''_{sib})$ where $p''_{sib}$ is
located on the line segment $p'p''$ such that \\
(i) $r''_{sib}$ is parallel to $r''$\\ 
(ii) $\cE (r'')$ is the same as $\cE (r''_{sib})$ \\
(iii) $d(r,r''_{sib}) = \delta$ (for $\delta < \epsilon'$).\\
Set $r', r'_{sib}$ and $r'', r''_{sib}$ as siblings and compute
current estimates of the shortest path to the endpoints of $e_i$ and $e_j$,
via the sibling rays.
\STATE
Push events to the heap which represents the weighted distance at which
ray $r'_{sib}$ intersects $e_i$ and the ray $r''_{sib}$ intersects $e_j$.
\end{algorithmic}
\end{algorithm}

\subsection{Procedure findCriticalPointOfEntry}
While propagating a ray it may happen that $r'$ and $r''$ are  two sibling rays
that originate from the same point and  $r'$ is incident to $e_i$ at an angle 
less than the critical angle while $r''$ incidents to $e_i$ at an angle greater than the critical angle.
Procedure {\it findCriticalPointOfEntry} handles this case and uses the procedures {\it wrapperToFindHitAtAngle} and {\it Find-Hit-At-Angle} together 
to find the ray $r$ which incidents at the critical angle on edge $e_i$ 
such that  the ray $r$ traverses across the same edge sequence as $r'$ and $r$ originates from $v$.

Since a vertex must exist between two critical points along a shortest path, we do not progress a ray $r$ striking an edge at an angle greater than or equal to critical angle whenever $r$ is originated from a critical segment (see procedure {\it wrapperToFindHitAtAngle}).
For the sake of completion, we have listed the 
procedure {\it Find-Hit-At-Angle} from \cite{Mitchell91} along with {\it wrapperToFindHitAtAngle}in the Appendix.

\begin{figure}
\center{
\subfigure[Sibling rays are separated by a vertex $v$]{\epsfysize=200pt \epsfbox{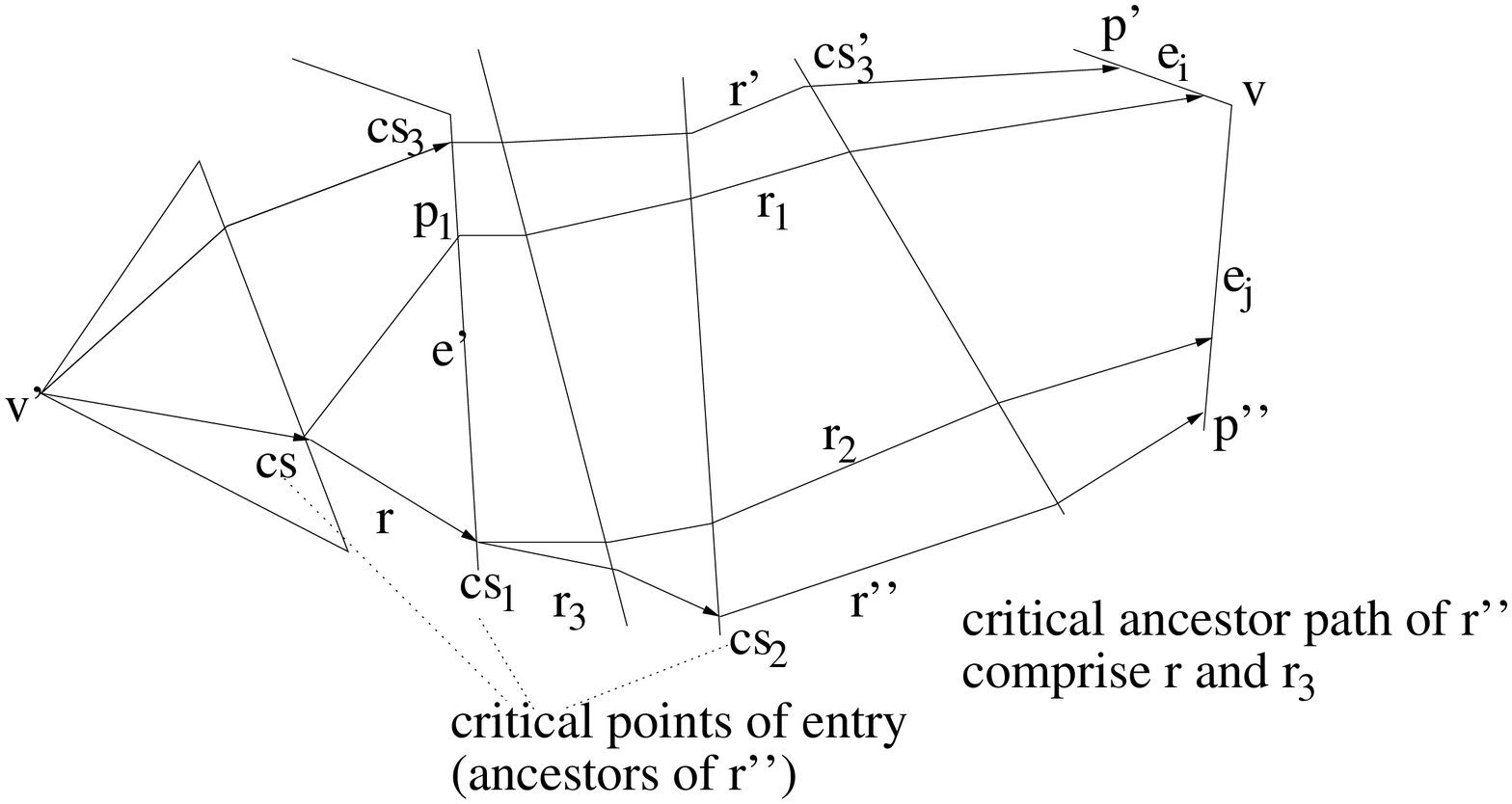}}
\subfigure[One sibling ray incidents at an angle less than the critical angle and the other incidents at an angle greater than or equal to the critical angle on $e_i$]{\epsfysize=200pt \epsfbox{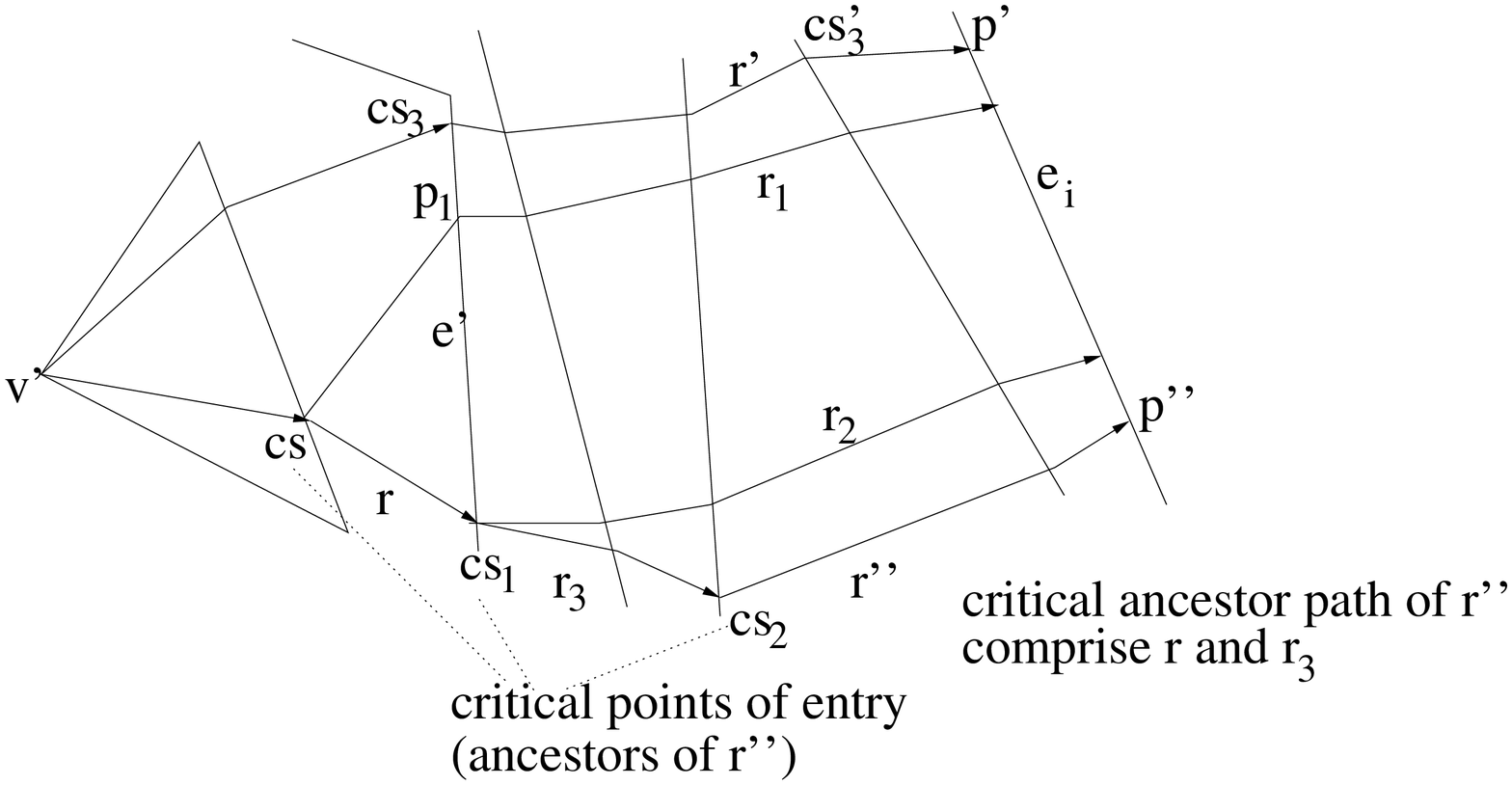}}
}
\caption{\label{fig:binsrchvrtsrc} Finding sibling rays using critical ancestor paths}
\end{figure}

\begin{figure}
\centerline{\epsfysize=200pt \epsfbox{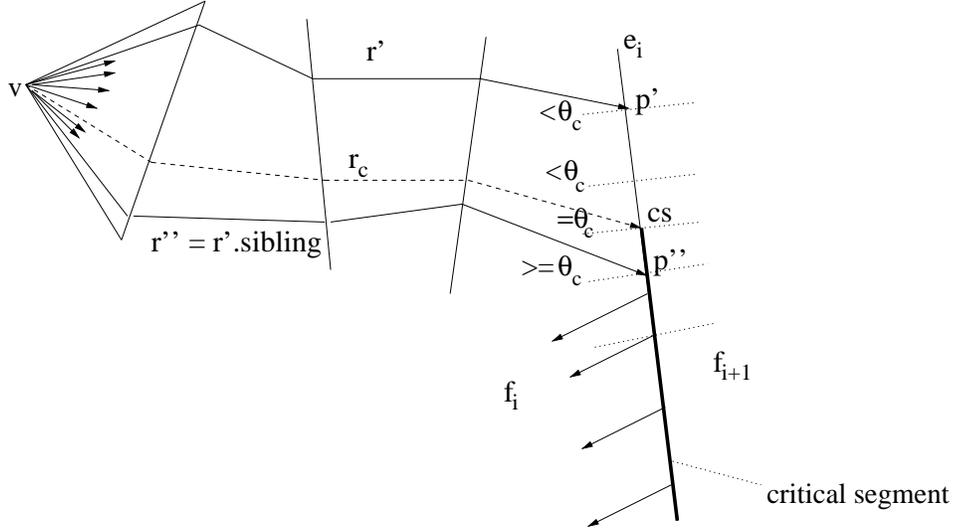}}
\caption{\label{fig:critangray} Finding the ray which incidents at critical angle on edge $e_i$}
\end{figure}

\begin{algorithm}
\caption{findCriticalPointOfEntry($r', e_i$)}
\begin{algorithmic}[1]
\REQUIRE{Ray $r'$ which incidents at point $p'$ located on edge $e_i=f_i \cap f_{i+1}$ with an angle of incidence less than or equal to critical, and, ray $r'.sibling=r''$ which incidents at point $p''$ of edge $e_i$ with an angle of incidence greater than or equal to critical. 
See Fig. \ref{fig:binsrchvrtsrc}(b).
}
\ENSURE{Finds a ray that originates from a critical source/vertex and incidents to a point $p$ located on the line segment $p'p''$ at an angle 
equal to the critical angle.}
\vspace*{.1in}
\STATE Let the sets $R_1$ and $R_2$ be defined as in procedure $findSplitRays$.
Do binary search on the critical sources defining the rays in $R_1$ and 
in $R_2$ and subsequently on the rays to find two rays $r_1$ and $r_2$ where $r_1$ incidents on $e_i$ at 
an angle less than or equal to critical and $r_2$ incidents on $e_i$ 
at an angle greater than or equal to critical.
\IF{$origin(r_1) = origin(r_2)$}
	\STATE wrapperToFindHitAtAngle($r_1, r_2, e_i, origin(r_1)$)
\ELSE
	\STATE Suppose $origin(r_2)$ is an intermediate vertex on the  critical ancestor path of $r''$ (the other case is symmetric).
And, let $r$ be the ray with origin same as $origin(r_1)$
that led to the critical source, $origin(r_2)$.
Then we invoke wrapperToFindHitAtAngle($r_1, r, e_i, origin(r_1)$)
\ENDIF
\end{algorithmic}
\end{algorithm}

\section{Correctness and Analysis}
\label{sect:weispcorranal}
\begin{lemma}
\label{type-1find}
The algorithm correctly determines sibling rays in $\cTray (u)$ for 
every 
source $u$, and computes an $\epsilon$-approximation to a Type-1 shortest path from $u$ to $v$ for any two  sources $u, v$.
\end{lemma}
\begin{proof}
Suppose that there exists a Type 1 shortest path between $u$ and $v$,
traversing a sequence of edges ${\cal E}$.
Lemma~\ref{thm:weisprefrnoncritical} shows that this path can be found
using sibling rays.
We show that the algorithm maintains sibling rays for edges in
${\cal E}$ that enables the determinations of traced rays
such that these rays can be used to find an $\epsilon$-approximate shortest path from $u$ to $v$.
First consider the case when the shortest path  lies in between two
sequence ${\cal E}$. These adjacent rays will be found from sibling rays
by the algorithm as 
shown below. The second case occurs when the shortest path lies 
outside the set of rays, $R$. In the second case, the vertex $v$ 
must be less than $d_w(v)\epsilon'$ from a ray (in fact a traced sibling ray) 
in $R$, since two points
at distance $d$ on adjacent rays, 
are separated by a distance of atmost $d\epsilon'$ 
(due to \ref{eq:xinbeta}  in Lemma~\ref{thm:weisprefrnoncritical}).
The algorithm computes the  estimate of the shortest distance to
a vertex $v$ whenever  sibling
rays are computed at every edge incident to $v$.

It thus suffices to show that sibling rays are computed correctly.
Initially the siblings are computed correctly in procedure {\it initiateVertexSource} (Lemma \ref{lemma:initiateVertexSource}) 
and each ray is refracted correctly in procedure {\it refractRay} 
(Lemma \ref{lemma:refractRay}).
We consider the other procedures whose correctness is not immediate from their descriptions.
Consider the procedure {\it findSplitRays} which computes the siblings
$r_1$ and $r_2$ closest to a  vertex $v$ at some stage in the algorithm.
The rays $r_1$ and $r_2$ are thus adjacent rays and the distance to $v$
via the rays $r_1$ and $r_2$ is updated.
Suppose both the rays $r_1$ and $r_2$ have the same origin.
Recall that Procedure {\it findSplitRays} finds the siblings and
due to (\ref{eq:xinbeta}) and (\ref{eq:13}), an $\epsilon$-approximation to the shortest distance from $origin(r')$ to $v$ is ensured.
Otherwise, $origin(r_1) \ne origin(r_2)$.
In this case procedure {\it findSplitRays} finds the siblings via
a binary search on the rays in the critical ancestor paths of two
siblings, say $r'$ and $r''$ such that rays $r_1$ and $r_2$  are part of
the set of rays lying in between $r'$ and $r''$.
Let $r$ be the ray which initiated $origin(r_2)$.
And let ray $r$ has as its origin $p=origin(r_1)$.
Also, we know that the ray $r_1$ is a successor to ray $r$ since it is adjacent to $r_2$.
Let $origin(r_2)$ be located on an edge $e'$.
Suppose $r_1$ intersects  $e'$ at point $p_1$.
See Fig. \ref{fig:binsrchvrtsrc}(a).
As described , by construction (procedure {\it initiateCriticalSource} and 
Lemma \ref{lemma:initiateCriticalSource}), ray $r_1$ is parallel to $r_2$, 
starting at $p_1$.
Due to (\ref{eq:xinbeta}), the distance between $p_1$ and $origin(r_2)$ is bounded.
These facts, together with (\ref{eq:13}), guarantee the $\epsilon$-approximation in computing the shortest distance from $origin(r')$ to $v$.

Note that {\it initiateCriticalSource} requires
the procedure {\it findCriticalPointOfEntry} whose correctness is clear from
the description.
\end{proof}

\begin{lemma}
\label{thm:weisprefrcritical-find}
Let $P$ be a Type-2  shortest path from a vertex $v$ to another vertex $w$ on $\PS$ with a critical segment ${\cal CS}$ in-between.
Then the algorithm determines a pair of traced rays  in $\cRay (v)$ 
that can approximate the shortest path, $P$. 
\end{lemma}
\begin{proof}
The proof of lemma~\ref{thm:weisprefrcritical}
shows  that sibling rays in $\cTray(v)$ and sibling rays
originating from ${\cal CS}$ will determine
the approximate shortest path.
The optimum path $P$ can be partitioned into two,
a path  $P_1$ from $v$ to $\cseg \in \cK$ and a path $P_2$ from $\cseg$ to $w$.
Let $e$ be the edge on a face $f$ such that the critical 
segment $\cseg$ lies  on $e$ and reflects rays back onto face $f$. 
Also, let the critical segment $\cseg$ have a critical point of entry $cs$.
We first show that a good approximation to the path from
$\cseg$ to $w$ can be found 
by showing that sibling rays are correctly maintained.
In fact, we need to find the sibling rays closest to 
vertex $w$ when the rays
from $\cseg$ strikes edges incident to $w$. 
Note that initially, two rays from  the end points of $\cseg$ are 
generated onto face $f$
correctly in procedure  {\it initiateCriticalSources} (Lemma \ref{lemma:initiateCriticalSource}), as the sibling rays. 
Subsequently, as these rays
strike edges of the faces they are refracted and siblings updated.. 
If the rays strike different edges $e_1$ and $e_2$, $e_1 \cap e_2 = w$, after
striking an edge $e'$ of a face comprising edges $e', e_1$ and $e_2$,
the appropriate siblings are computed
in procedure {\it findCriticallyReflectedRay}.
Also the shortest distance path to $w$ starting at the critical segment
$\cseg$ is computed. This path has a first segment that 
is parallel to the sibling rays, $r'$ and $r''$, that are 
initially generated in 
procedure {\it initiateCriticalSource} and maintained during the progress of
the propagation. Furthermore it has its 
origin, $p=origin(r)$ located in-between $origin(r')$ and $origin(r'')$.
Such a path is found in Step 1 of
procedure {\it findCriticallyReflectedRay}.
If the sibling rays do not split but strike one edge, say $e_1$, the
distance via the closest sibling to $w$ is computed as required
in lemma~\ref{thm:weisprefrcritical}.

Finally, a Type-1 path between $u$ and $cs$ is found 
correctly by the algorithm as 
specified in lemma~\ref{type-1find}
\end{proof}

\ignore{
\begin{theorem}
The algorithm computes a shortest path which is an $\epsilon$-approximation to an exact shortest path from $s$ to $t$ 
\end{theorem}
\begin{proof}
This follows from Lemmas \ref{thm:weisprefrnoncritical}, \ref{thm:type1andtype2}, and \ref{thm:weisprefrcritical-find}.
For any point $p$ located in the region with the weight $\alpha$, consider two rays $r', r''$ adjacent to $p$ such that $r', r''$ originated from the same source. 
Then for two points $p_1, p_2$, chosen on rays $r', r''$ respectively such that the line segment $p_1p_2$ passes through $p$ and $p_1, p_2$ are at the same weighted distance $d$ from $s$. 
As proved in Lemmas \ref{thm:weisprefrnoncritical} and \ref{thm:weisprefrcritical}, $\alpha \Vert p_1p_2 \Vert$ is at most $d\epsilon$.
Hence, the weighted distance from $s$ to $p$ is at most $d(1+\epsilon)$.
\end{proof}
}

\begin{theorem}
\label{} 
The algorithm computes an $\epsilon$-approximation to an exact shortest path from $s$ to $t$ 
in $O(n^6\lg(\frac{K}{\epsilon'})+n^4\lg(\frac{W}{w\epsilon'}))$ time complexity.
\end{theorem}
\vspace*{-.10in}
\begin{proof}
The correctness follows from the previous two lemmas
and by splitting the optimal path from $s$ to $t$ into parts, each of which can be classified as either a Type 1 or a Type 2 path.

We bound the time complexity by  
\begin{enumerate}

\item 
Computing the time required to determine sibling rays, i.e. rays that are traced by the algorithm in {\it  findSplitRays}.

\item
Computing the time required in {\it findCriticalPointofEntry}.
\end{enumerate}
The other procedures called in {\it weightedEuclideanSP} will be shown to have lower order of complexity.
The introduction of sibling rays is done in procedures {\it findSplitRays} and  {\it findCriticallyReflectedRay}.  When the algorithm determines new siblings, it is required to do binary search over $\cRay (v) , v \in C$.
The total number of rays in $\cRay (v) , v \in {\cal C}$ are $O(\frac{W}{w\epsilon'}(\frac{K}{\epsilon'})^{n^2})$.
It takes $O(\lg(\frac{W}{w\epsilon'}(\frac{K}{\epsilon'})^{n^2}))$ steps to do binary search over these rays where each step involves tracing the ray forward and computing the intersection with the  edge sequence to determine the point at which the ray strikes one of the edges $e_i$ or edge $e_j$.   
Since an edge can appear in an edge sequence $O(n)$ times (Lemma \ref{lem:edgeseqlen}), the length of an edge sequence is $O(n^2)$.  
This also implies that the size of the sets $R_1$ and $R_2$ (which are defined in procedure {\it findSplitRays}), are $O(n^2)$. 
Consider the critical ancestor path $P$ such that the origin of $P$ is $v$ and $P$ comprises the maximum possible number of critical points. 
Let $S$ be the set consisting of all the critical points of entry in path $P$ together with $v$.
Note that if we do a binary search over the set of rays in $\cRay(v') , v' \in S$, we need not  do binary search over $\cRay(v'') , v'' \in S$. 
It then remains to determine the vertex $v' \in S$ where binary search is required.
This can be done by considering the vertices on the critical ancestor paths, which are of size $O(n^2)$.
Therefore, the work involved for each invocation of {\it findSplitRays} is $O(n^2\lg(\frac{W}{w\epsilon'}(\frac{K}{\epsilon'})^{n^2}))$ and since this may be required for all pairs of vertices in $\PS$, the total work is $O(n^6\lg(\frac{K}{\epsilon'})+n^4\lg(\frac{W}{w\epsilon'}))$.

With similar reasoning, the total work due to binary searches in the procedure {\it findCriticalPointOfEntry} takes $O(n^6\lg(\frac{K}{\epsilon'})+n^4\lg(\frac{W}{w\epsilon'}))$ time.

We now consider the time required for other procedures.
The main algorithm also invokes the procedures {\it initiateVertexSource}, {\it initiateCriticalSource} and  {\it refractRay}.
The first two procedures use binary search to find sibling rays. 
This is done only $O(1)$ times for every procedure call. 
And these procedure are called $O(n^3)$ times since there are $O(n^2)$ critical points of entry (Lemma \ref{lem:critsources}) and  there can be $O(n^3)$ possible combinations of vertex and critical points of entry.
Next consider the time for procedure {\it refractRay}. 
Each invocation of this procedure results in processing a ray and its sibling.
The complexity of the calls to procedure {\it findSplitRays} during the procedure is detailed above. 
For the complexity of the remainder operations we note that the procedure {\it refractRay} requires $O(1)$ work when called for each strike of sibling rays onto an edge of a face.
The number of sibling rays is bounded by $O(n)$ for each source $v$ since a pair of sibling rays can be associated with a vertex which caused the set of rays in $\cRay(v)$ to split into two. 
This bounds the number of calls to {\it refractRay} to $O(n^4)$.
Similarly, since there are $O(n^2)$ critical points of entry (Lemma \ref{lem:critsources}), 
there can 
be $O(n^3)$ possible combinations of vertex and critical points of entry such that it is 
required to find the sibling rays in procedure {\it findCriticallyReflectedRay}.
Hence the procedure {\it Find-Hit-At-Angle} is invoked $O(n^3)$ times, whereas 
each invocation takes $O(n^2)$ time i.e., the total time spent in 
procedure {\it Find-Hit-At-Angle} 
during the entire algorithm is $O(n^5)$.
Since an edge sequence is of length $O(n^2)$, the work involved in tracing the ray incident to a vertex from a critical segment takes $O(n^2)$ time (line 2 of procedure {\it findCriticallyReflectedRay}).
As there are $O(n)$ vertices and $O(n^2)$ critical segments (Lemma \ref{lem:critsources}), the total work involved takes $O(n^5)$, which is subsumed by 
the complexity of other procedures.
\end{proof}

\section{Conclusions}
\label{sect:weispconclu}
This paper presents an algorithm with $O(n^6\lg(\frac{K}{\epsilon'})+n^4\lg(\frac{W}{w\epsilon'}))$ time complexity to find the shortest path from $s$ to $t$ in a weighted subdivision.
This algorithm is of better time complexity than the existing polynomial time algorithm and 
it is better than the pseudo-polynomial time algorithms 
whenever $\frac{N^2}{\sqrt{\epsilon}} > n^5$.

\bibliography{weisp-rinkulu10}

\begin{thebibliography}{1}

\bibitem{Agarwal02}
P.~Agarwal, S.~Har-Peled, and M.~Karia.
\newblock {Computing Approximate Shortest Paths on Convex Polytopes}.
\newblock {\em Algorithmica}, 33(2):{227--242}, 2002.

\bibitem{Aleks00}
L.~Aleksandrov, A.~Maheshwari, and J.-R. Sack.
\newblock {Approximation Algorithms for Geometric Shortest Path Problems}.
\newblock In {\em {Proceedings of the ACM Symposium on Theory of Computing}},
  pages 286--295, 2000.

\bibitem{Aleks05}
L.~Aleksandrov, A.~Maheshwari, and J.-R. Sack.
\newblock {Determining Approximate Shortest Paths on Weighted Polyhedral
  Surfaces}.
\newblock {\em Journal of the ACM}, 52(1):25--53, 2005.

\bibitem{Cheng08}
S.-W. Cheng, H.-S. Na, A.~Vigneron, and Y.~Wang.
\newblock {Approximate Shortest Paths in Anisotropic Regions}.
\newblock {\em SIAM Journal on Computing}, 38(3):802--824, 2008.

\bibitem{Lanthier01}
M.~Lanthier, A.~Maheshwari, and J.-R. Sack.
\newblock {Approximating Weighted Shortest Paths on Polyhedral Surfaces}.
\newblock {\em Algorithmica}, 30(4):527--562, 2001.

\bibitem{Mata97}
C.~Mata and J.~Mitchell.
\newblock {A New Algorithm for Computing Shortest Paths in Weighted Planar
  Subdivisions (Extended Abstract)}.
\newblock In {\em Proceedings of the ACM Symposium on Computational Geometry},
  pages 264--273, 1997.

\bibitem{Mitchell91}
J.S.B. Mitchell and Christos~H. Papadimitriou.
\newblock {The Weighted Region Problem: Finding Shortest Paths Through a
  Weighted Planar Subdivision}.
\newblock {\em Journal of the ACM}, 38(1):18--73, 1991.

\bibitem{Sun01}
Z.~Sun and J.~H. Reif.
\newblock {BUSHWHACK: An Approximation Algorithm for Minimal Paths Through
  Pseudo-Euclidean Spaces}.
\newblock In {\em Proceedings of the International Symposium on Algorithms and
  Computation}, pages 160--171, 2001.

\bibitem{Reif00}
Z.~Sun and J.~H. Reif.
\newblock {On finding approximate optimal paths in weighted regions}.
\newblock {\em Journal of Algorithms}, 58(1):1--32, 2006.

\end{thebibliography}

\pagebreak

\section{Appendix}
\label{sect:appendix}
\begin{algorithm}
\caption{wrapperToFindHitAtAngle($r', r'', e_i, v$)}
\begin{algorithmic}[1]
\REQUIRE{
Let the critical angle associated with edge $e_i$ be $\theta_c$.
The ray $r'$ incidents at point $p'$ located on the edge $e_i=f_i \cap f_{i+1}$ with an angle less than or equal to $\theta_c$ and the ray $r''$ incidents at point $p''$ located on the edge $e_i$ with an angle greater than or equal to $\theta_c$, for $origin(r') = origin(r'') = v$.  
See Fig. \ref{fig:critangray}.
}
\ENSURE{This is a wrapper to the procedure {\it Find-Hit-At-Angle}.} 
\vspace*{.1in}
\IF{$r_1.sourcetype$ != critseg}
	\STATE Define the interval $I$ as $[p', p'']$.
	\STATE $cs \leftarrow$ Find-Hit-At-Angle($\theta_c, I, e_i, v, r'$)
	\STATE 
        Let $r_c$ be the ray causing $cs$.  
        Set $r_1, r_c$ as siblings.
	Let $\theta$ be the angle at which $r'$ refracts from $p'$, and, let $seg$ be the critical segment due to $cs$.  
        Push new event corresponding to the critical point of entry $cs$ into the heap,
 with $seg, \theta, \theta_c$ as  associated  data.
\ENDIF
\end{algorithmic}
\end{algorithm}

\begin{algorithm}
\caption{Find-Hit-At-Angle($\theta_c, I, e_i, v, r'$) (procedure from \cite{Mitchell91})}
\begin{algorithmic}[1]
\REQUIRE{an interval $I$, critical angle of reflection $\theta_c$, edge $e_i$ such that $e_i \cap I = I$, ray $r'$ which intersects an endpoint of $I$, point $v$ which is the origin of sibling rays incident to endpoints of $I$}
\ENSURE{Find the point in $I$ (if there is one) such that the refraction path from the root of $I$ through the last edge sequence will hit the point at the angle of incidence $\theta$.
If the function succeeds in finding the desired path, it returns the critical point of entry in the interval $I$;
else, returns NIL.
}
\vspace*{.1in}
\STATE 
Starting with an angle of incidence $\theta_c$ at edge $e_i$, determine the angles of incidence at each edge in the last edge sequence of $r'$.
This can be done by following back pointers to predecessors of intervals and using the fact that the angle of incidence at $e_i$ uniquely specifies the angle of refraction from $e_{i-1}$, which in turn uniquely specifies the angle of incidence at $e_{i-1}$.
It may be impossible at some stage to achieve the desired angle of refraction, due to the constraint that the angle of incidence be less (in absolute value) than the critical angle; in this case, we stop and return NIL.
This process will eventually give the angle $\beta$ such that a refraction path that starts from $v$ at an angle $\beta$ will refract through the last edge sequence or $r'$ and strike $e_i$ at an angle of incidence $\theta_c$.
\STATE 
Now do ray tracing from $v$ at an angle $\beta$, following the last edge sequence or $r'$.
If the refraction path leaves the sequence, then stop and return NIL.
Otherwise, the refraction path will eventually cross edge $e_i$ at some point $x$ (and we know by construction that it will strike at point $x$ with the desired angle of incidence $\theta_c$).
If $x \in I$, then return $x$; 
otherwise, return NIL.
\end{algorithmic}
\end{algorithm}

\end{document}